\documentclass[5p]{elsarticle}

\usepackage[ruled, vlined, linesnumbered]{algorithm2e}
\usepackage{verbatim}
\usepackage{amsmath}
\usepackage{url}
\usepackage{multirow}
\usepackage{graphicx}
\usepackage[hidelinks]{hyperref}
\usepackage{picins}
\usepackage{amsthm}

\newtheorem{theorem}{Theorem}[section]
\newcommand{\argmin}{\operatornamewithlimits{argmin}}
\newcommand{\argmax}{\operatornamewithlimits{argmax}}

\begin{document}

\begin{frontmatter}
\title{CeMon: A Cost-effective Flow Monitoring System in Software Defined Networks}

\author[hkust]{Zhiyang Su}
\ead{zsuab@cse.ust.hk}
\author[hkust]{Ting Wang}
\ead{twangah@cse.ust.hk}
\author[hkust]{Yu Xia}
\ead{rainsia@cse.ust.hk}
\author[hkust]{Mounir Hamdi}
\ead{hamdi@cse.ust.hk}

\address[hkust]{The Hong Kong University of Science and Technology, Clear Water Bay, Kowloon, Hong Kong}

\begin{abstract}
  Network monitoring and measurement are crucial in network management to facilitate quality of service routing and performance evaluation. Software Defined Networking (SDN) makes network management easier by separating the control plane and data plane. Network monitoring in SDN is relatively light-weight since operators only need to install a monitoring module into the controller.
Active monitoring techniques usually introduce extra overhead into the network. The state-of-the-art approaches utilize sampling, aggregation and passive measurement techniques to reduce measurement overhead. However, little work has focused on reducing the communication cost of network monitoring. Moreover, most of the existing approaches select polling switch nodes by sub-optimal local heuristics.

Inspired by the visibility and central control of SDN, we propose CeMon, a generic low-cost high-accuracy monitoring system that supports various network management tasks. We first propose a Maximum Coverage Polling Scheme (MCPS) to optimize the polling cost for all active flows. The problem is formulated as a weighted set cover problem which is proved to be NP-hard. Heuristics are presented to obtain the polling scheme efficiently and handle traffic dynamics practically. In order to balance the cost and flexibility, an Adaptive Fine-grained Polling Scheme (AFPS) is proposed as a complementary method to implement flow level measurement tasks. Three sampling algorithms are further proposed to eliminate measurement overhead while maintain high accuracy. 
Both emulation and simulation results show that MCPS reduces more than $50\%$ of the communication cost with negligible loss of accuracy for different topologies and traffics. We also use real traces to demonstrate that AFPS reduces the cost by up to $20\%$ with only $1.5\%$ loss in accuracy.
\end{abstract}


\end{frontmatter}


\section{Introduction}
Network monitoring is a common task in network management. Flow-based measurements plays an important role in network management. Low-cost, timely and accurate flow statistics collection is crucial for different management tasks such as traffic engineering, accounting and intelligent routing. For example, many data centers collects flow statistics in the orders of second to dynamically schedule flow routing~\cite{hedera, helios}.

Traditional network monitoring techniques such as NetFlow~\cite{netflow} and sFlow~\cite{sflow} support flow-based measurement tasks. However, they have a higher deployment cost and consume much resource~\cite{opensketch}. For example, the deployment of NetFlow consists of setting up collectors, analyzers and other services. Moreover, enabling NetFlow in the routers may degrade the packet forwarding performance~\cite{reformulateplacement}. Besides, these passive measurement techniques requires full access to the network devices which raises privacy and security issues. 

By separating the control plane and the data plane, SDN provides unprecedented flexibility to conduct network measurement. The fundamental primitive for existing software defined measurement frameworks is flow statistics collection~\cite{opensketch, progme, dcm, onlineaggregate}. If a flow has corresponding forwarding rules in a switch, it is regarded as an active flow. The controller is able to track all the active flows by passively receiving flow arrive and flow expired notifications from the switches. Monitoring flow statistics in SDN is relatively light-weight and easy to implement: the central controller maintains the whole network states, and is able to poll flow statistics from any switch periodically.

Recent pull-based measurement proposals such as OpenTM~\cite{opentm} obtain flow statistics based on a per-flow querying strategy. If the number of active flows is large, the extra communication cost for each flow cannot be neglected. Due to the limited bandwidth between the controller and the switches, the monitoring traffic is likely to result in a bandwidth bottleneck~\cite{mahout}. The situation becomes worse for in-band SDN deployment when monitoring and routing traffic shares bandwidth along the same link. In contrast, FlowSense~\cite{flowsense} infers link utilization by passively capturing the flow arrival and expiration messages with zero overhead. However, FlowSense calculates the link utilization only at discrete points in time after the flow expires. This limitation cannot meet the real-time requirement, neither can the accuracy of the results be guaranteed. Therefore, the key challenge for is how to design a high-accuracy flow statistics collection scheme at minimum polling overhead. However, eliminating the bandwidth consumption for measurement traffic has not been studied so far.

Inspired by the global view of SDN and existing software defined measurement frameworks~\cite{opensketch, opentm, flowsense}, we propose a novel flow statistics collection system CeMon, a low-cost high-accuracy system that collects the flow statistics across the network in a timely fashion. The design of CeMon is based on the observation that per-flow querying strategy is sub-optimal as it lacks globally optimization to choose the polling switches. By aggregating the flow statistics collection queries and optimizing the polling frequency, CeMon significantly reduces the flow statistics collection cost. Such optimization is of great importance for network monitoring, especially in a high-accuracy monitoring scenario that requires real-time statistics collection~\cite{hedera, helios}. 

CeMon is generic, efficient and accurate. First, CeMon is able to cooperate with other software defined measurement frameworks. This property is guaranteed by the implementation of the flow statistics collection primitive. Working between the physical network and the measurement applications, other frameworks are able to invoke CeMon to collect flow statistics at minimum monitoring overhead. Second, thanks to the proposed heuristic, CeMon is able to efficiently generate the polling scheme within two seconds for a huge number of active flows. Finally, we prove the performance and the accuracy bound of our heuristic. Extensive experimental results demonstrate that CeMon reduces up to $50\%$ monitoring overhead with negligible loss in accuracy.

The primary contributions of this paper are as follows.

\begin{itemize}
\item
  We provide a general framework (Section~\ref{sec_design}) to facilitate various monitoring tasks such as link utilization, traffic matrix estimation, anomaly detection and so on. It is a shim layer between the physical network and measurement applications, which is compatible with most of current software defined measurement frameworks and significantly reduces the cost to fetching flow statistics.
\item
  We propose a Maximum Coverage Polling Scheme (MCPS) (Section~\ref{sec_mcps}) which globally optimizes the polling cost. Furthermore, MCPS is generic and  can be applied to out-of-band deployment, in-band deployment and multiple controllers.
\item
  We propose an Adaptive Fine-grained Polling Scheme (AFPS) (Section~\ref{sec_afps}) which supports flow level measurements at low-cost. AFPS leverages different adaptive algorithms to dynamically adjust polling frequency, which eliminates the measurement overhead with negligible loss of accuracy.
\end{itemize}

The rest of this paper is structured as follows. Section~\ref{sec_design} introduces the background and presents the architecture of CeMon. Section~\ref{sec_mcps} formulates the maximum coverage polling problem, proposes heuristics to generate solution efficiently and to handle flow dynamics. Section~\ref{sec_afps} presents a fine-grained flow level measurement framework and proposes adaptive polling algorithms to support various measurement tasks. Section~\ref{sec_evaluation} elaborates on the performance of MCPS and AFPS by real packet traces. Finally, Section~\ref{sec_relatedwork} summarizes related work and Section~\ref{sec_conclusion} concludes the paper.

\section{System Design} \label{sec_design}
In this section, we first introduce the backgrounds of SDN and OpenFlow. The architecture of CeMon and its workflow are presented thereafter.

\subsection{Background}
OpenFlow~\cite{openflow} is an open standard of SDN. Currently, OpenFlow-based SDN is widely used in both industry and academia. OpenFlow is the de facto standard communication interface between the control plane and the data plane. The controller is able to add, remove and modify rules in the switches to operate routing and monitoring actions. When the first packet of a new flow arrives at the edge switch, a table miss is raised and the packet header will be forwarded to the controller. The controller processes the packet header and takes further actions such as setting up the routing path. According to the OpenFlow specification 1.0~\cite{openflowspec10}, the minimum lengths of flow statistics request and reply messages on wire are 122 bytes and 174 bytes respectively.
\begin{figure}[!t]
  \centering
  \includegraphics[width=0.8\linewidth]{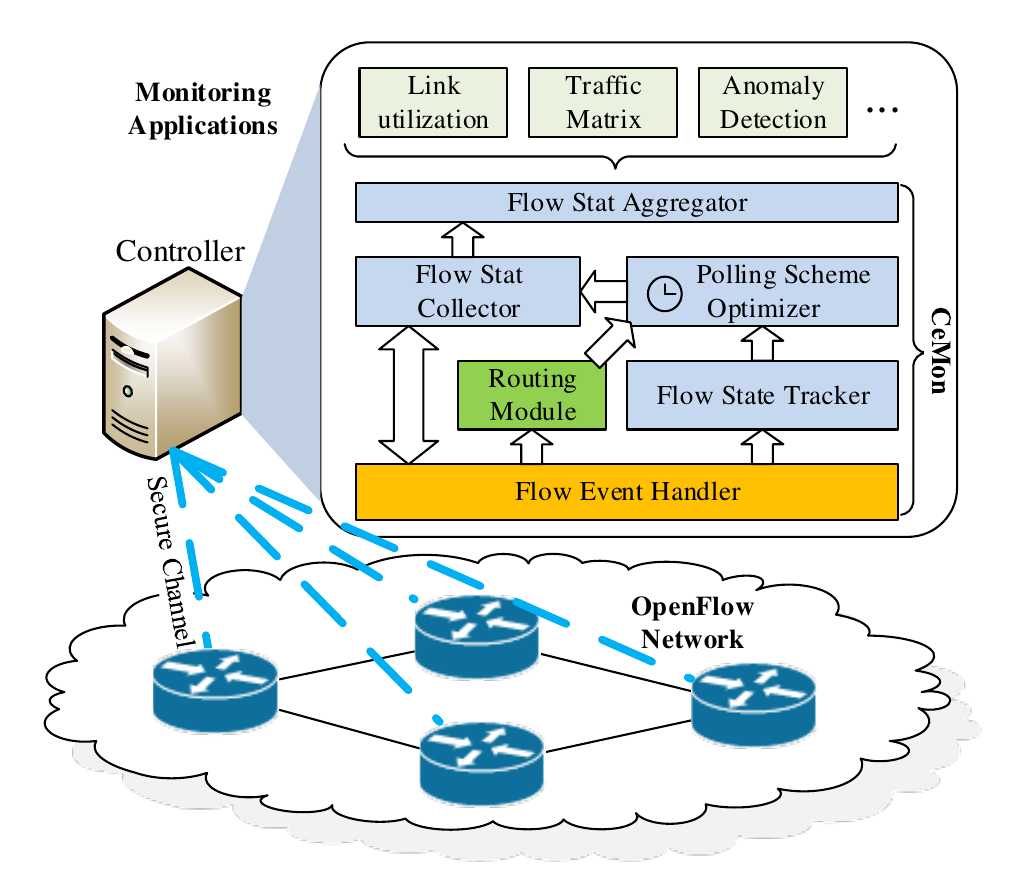}
  \caption{CeMon architecture.}
  \label{fig_architecture}
\end{figure}

The deployment of SDN-based networks can be categorized into two groups: out-of-band deployment and in-band deployment~\cite{openflowspec10, inbandfastrecovery}. The out-of-band deployment transmits the control messages in a dedicated control network. To the contrary, the in-band deployment transmits the control messages and data messages through the same network.
Since the out-of-band deployment isolates the control and data messages, it provides better performance isolation, fault tolerance and privacy. However, it is worth noting that the deployment cost of out-of-band deployment is much higher as a dedicated control network is needed. Therefore, the in-band deployment is preferable in practice. The routing scheme of the control messages for the in-band deployment is determined by the network operator. A separated VLAN can be configured to deliver the control messages.

\subsection{CeMon Architecture} \label{sec_architecture}
Basically, the monitoring task in SDN is accomplished by the controller which is connected to all the switches through a secure channel, which is usually a TCP connection between the controller and the switch. The controller collects real-time flow statistics from the corresponding switches, merges the raw data and passes the results to the upper-layer applications.
\begin{figure}[!t]
  \centering
  \includegraphics[width=\linewidth]{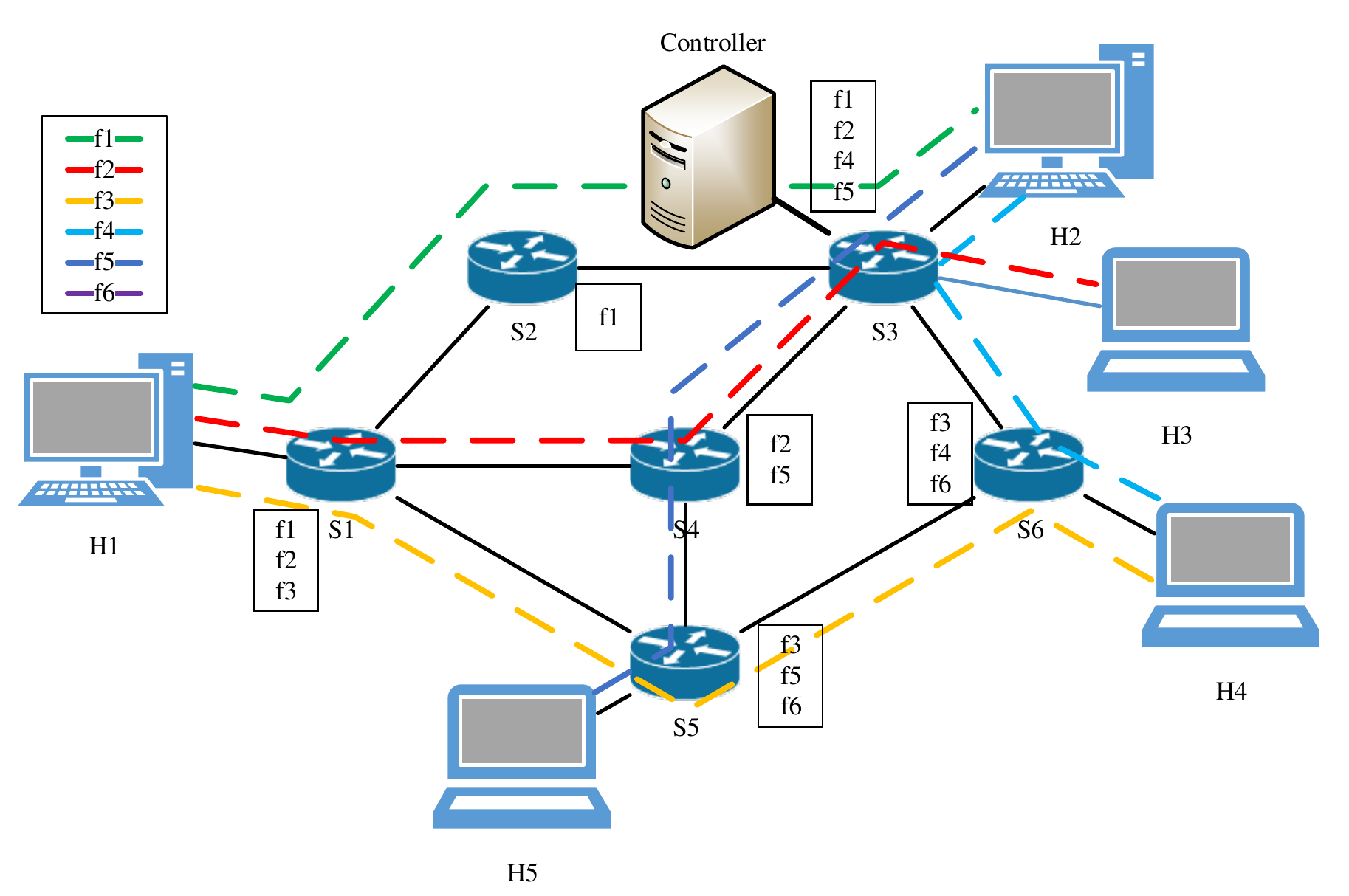}
  \caption{Flow statistics collection example. The network consists of six switches and five hosts. There are six active flows: $f1: H1-H2; f2: H1-H3; f3: H1-H4; f4: H2-H4; f5: H2-H5; f6: H4-H5$. The passing flows of each switch are marked in rectangles. The controller is attached to $S3$ for in-band deployment.}
  \label{fig_motivation}
\end{figure}

We describe the architecture of CeMon in Figure~\ref{fig_architecture}. In general, there are three layers: OpenFlow network layer, CeMon core layer and monitoring application layer. The OpenFlow network layer consists of underlying low-level network devices and keeps connections between the controller and the switches. The CeMon core layer is the heart of the monitoring framework. The flow event handler receives the flow arrival/expiration messages from switches and forwards them to the routing module and the flow state tracker. While the routing module calculates the routing path in terms of the policy specified by the administrator, the flow state tracker maintains the active flows in the network. The routing module and the flow state tracker report the active flow sets and their corresponding routing paths to the polling scheme optimizer respectively. Based on the above information, the polling scheme optimizer computes a cost-effective polling scheme and forwards it to the flow stat collector. The flow stat collector takes the responsibility of polling the flow statistics from the switches and handles the reply. Finally, the flow stat aggregator gathers the raw flow statistics and provides interfaces for the upper monitoring applications. The monitoring application layer is a collection of various tasks such as link utilization, traffic matrix estimation and anomaly detection. The CeMon core layer interacts with the OpenFlow network layer through OpenFlow protocol. The CeMon core layer provides an API to the monitoring application layer to return the statistics of a set of flows, which are specified in the API parameter. Essentially, the CeMon components are controller modules, which interact with each other by function calls.

The architecture of CeMon is compatible with other existing software defined measurement frameworks~\cite{opensketch, dream, opentm, dcm}. Since all these proposals are flow-based measurements, the final stage of the measurement is flow statistics collection. Therefore, these architectures can leverage the optimized polling scheme by CeMon to reduce the monitoring overhead. For example, DREAM~\cite{dream} periodically retrieves flow counters from the switches and passes them to task objects. Because the stage of fetching counters is independent of the task assignment, CeMon can be easily integrated to DREAM by modifying the fetching counter function from the per-flow querying to CeMon polling scheme.

Similarly, since CeMon implements the fetching counter primitive, many measurement tasks can be built on top of CeMon such as link utilization~\cite{payless}, flow size distribution~\cite{flowsizedistribution} and anomaly detection~\cite{adaptiveflowcounting}. For example, to get the utilization of a link, CeMon keeps track of all the active flows that pass the link and periodically polls their statistics. By adding up the utilization of each flow, CeMon constructs the utilization of this link.

As stated, the key challenge for flow statistics collection is the generation of a cost-effective polling scheme. In the following sections, we present two novel polling schemes MCPS and AFPS respectively, where MCPS focuses on gathering all flow statistics in an efficient way, while AFPS aims at polling a small number of flow statistics with high flexibility. 
\begin{figure}[!t]
  \includegraphics[width=\linewidth]{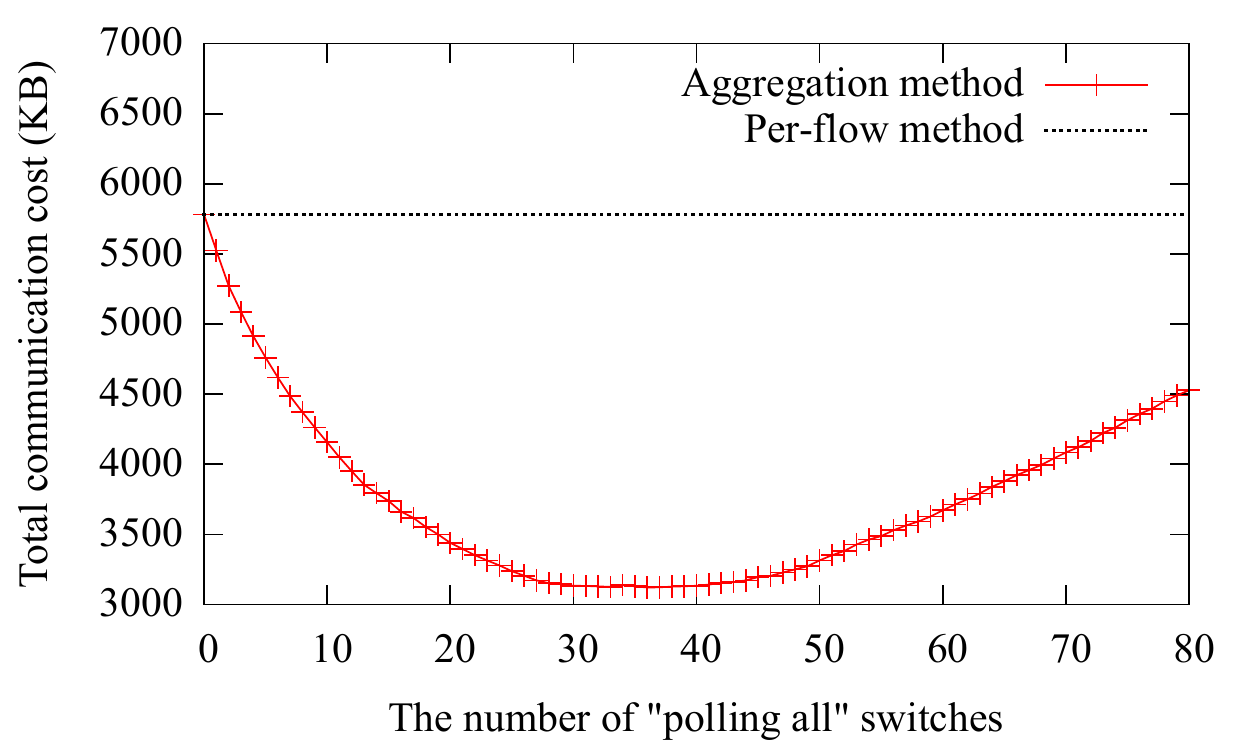}
  \caption{The number of ``polling all'' switches vs. the communication cost in a network with $100$ switches and $20000$ active flows.}
  \label{fig_setcover}
\end{figure}

\section{Maximum Coverage Polling Scheme} \label{sec_mcps}
In this section, we elaborate on the maximum coverage polling scheme in detail. We first motivate the polling all flow statistics problem, then we analysis the communication cost of SDN monitoring systems for both out-of-band and in-band deployments. Formal problem formulations are given thereafter and our solutions are applicable for different deployments and multiple controllers. Due to the computation complexity of the problem, heuristics are proposed to efficiently produce the polling scheme. Finally, we discuss flow dynamics and present corresponding heuristics to tackle this issue.

\subsection{Overview} \label{sec_mcpsoverview}
There are two ways to poll the statistics of a flow from the switches: one is a polling single query which only fetches the counter of the flow, another is a polling all query which fetches all the flow statistics in a switch. Define the communication cost as the sum of the lengths of polling request and reply messages. To collect the statistics of all the active flows, it is promising to combine the polling single query and the polling all query to cover all the flows at minimum communication cost. OpenFlow specification 1.0~\cite{openflowspec10} defines a match structure to identify one flow entry or a group of flow entries. A match structure consists of many fields to match flows such as input switch port, Ethernet source/destination address and source/destination IP. However, it is impractical to select an arbitrary subset of flows with ``segmented'' fields due to the limited expression of a single match structure.

Figure~\ref{fig_motivation} illustrates the polling all flow statistics problem for out-of-band deployment. Consider the four flows passing $S3$, assume the source and destination of these flows are: $f1:(H1,H2)$; $f2:(H1,H3)$; $f4:(H2,H4)$; $f5:(H5, H2)$. Intuitively, for the polling all request and reply, the lengths of the messages are in proportion to the requested number of flows. The total cost of each polling scheme can be measured by Wireshark. The specification specifies the polling request message length: to poll a single flow, the request and the reply lengths are 122 and 174 bytes, respectively. Then, we enumerate all the possible polling schemes and find the most cost-effective one. Intuitively, we prefer to choose the switches which have more active flows. The reason is that we can use less polling requests to obtain more flow statistics. In this example, the optimal solution is querying $S3$ and $S6$, with two requests of a total communication cost of $C_{opt}=1072$ bytes, where polling $S3$ and $S6$ with cost of $488$ and $584$ respectively. Compared with the cost of the per-flow querying with six requests of a cost of $C_{\text{per-flow}}=1776$ bytes, we save about $\frac{1776-1072}{1776}=39.6\%$ of the communication cost. Detailed modeling of the message length is in Section~\ref{sec_formulation}.

The case becomes worse for in-band deployment: the measurement and data traffic shares bandwidth, proactively fetching counters with high frequency notably impacts the efficiency of data transmission. Also, it is relatively complex to compute the communication cost compared with the out-of-band deployment. The hops from the polling switch to the controller should be taken into account. Suppose the controller is attached to $S3$ and we use the shortest path as the control message routing algorithm in Figure~\ref{fig_motivation}. For a single polling, we can query any switch along the flow path. However, the polling cost is different for the in-band deployment as the switches are of different distances to the controller. Without loss of generality, we employ random choosing and minimum cost choosing strategies to compare with our approach. Similar to the out-of-band case, the optimal solution is querying $S3$ and $S6$, which yields a cost of $C_{opt}=1560$ bytes. Compared with a random per-flow querying $C_{\text{randomper-flow}}=4144$ bytes ($f1,f2,f3:S1$; $f4,f5:S3$; $f6:S5$), the minimum cost querying $C_{\text{minimumcostper-flow}}=2368$ bytes ($f1,f2,f4,f5:S3$; $f3,f6:S6$), we save $\frac{4144-1560}{4144}=62.4\%$ and $\frac{2368-1560}{2368}=34.1\%$ of the communication cost respectively.

Essentially, polling flow statistics from one switch as much as possible is a sort of aggregation. However, if this ``polling all'' strategy is excessively used, it brings extra overhead due to repeatedly gathering the same flow statistics from different switches. To further explore the problem, we use a simple greedy algorithm which chooses switches that cover the most number of uncovered flows to collect all the flow statistics. Figure~\ref{fig_setcover} illustrates the increment of total communication cost with the number of ``polling all'' switches varying from $0$ to $80$. The dashed line is the total communication cost of the per-flow method for comparison. For the aggregation method, there has been a steady fall before the number of ``polling all'' switches reaches $30$. After reaching the bottom, the total communication cost rises gradually until all the active flows have been covered. This observation motivates us to globally optimize the polling strategy which minimizes the monitoring overhead.

\subsection{Problem Formulation} \label{sec_formulation}
As mentioned in Section~\ref{sec_mcpsoverview}, we can poll flow statistics from a switch by two strategies: (1) exact match of one flow and (2) wildcarding all fields to collect all flows. The benefits of the latter strategy are that we reduce the number of request messages and repeated reply headers. On the other hand, excessive usage of the second strategy imposes extra communication cost as there are overlap flow statistics in the replies. Therefore, the problem can be formulated as an optimization problem whose objective is to minimize the communication cost.

The target network is an undirected graph $G=(V,E)$, where $V=\{v_1,v_2,\ldots,v_n\}$ is the set of switches and $E$ represents the set of links. Therefore, $n=|V|$ is the number of switches in the network. There are $m$ active flows in the network $F=\{f_1, f_2, \ldots f_m\}$ (called the universe), where each element $f_i,i=1,2,\ldots, m$ corresponds to a sequence of switches $P_i$ that represents the flow routing path with length $l$: $P_i=(v_{j_1},v_{j_2}, \ldots, v_{j_l}), j_q \in [1, n],q \in [1,l]$. Let set $s_i$ denote the active flows in switch $v_i$. Then, $s_i$ can be generated by $F$ and $P_i$. Assume the polling all set $S=\{s_1,s_2,\ldots, s_n\}$, the active flow number in switch $v_i$ is $|s_i|$. Let $l_{req}$ denote the length of the flow statistics request message, $l_{rh}$ denote the length of flow statistics reply message header, $l_{sf}$ denote the length of reply message body of a single flow entry. For a flow statistics reply message with $n$ entries, the whole reply message length $l_{reply}(n)$ is a linear function of $n$\footnote{According to the OpenFlow specification~\cite{openflowspec10}, $l_{req}=122$ bytes, $l_{replyheader}=78$ bytes, $l_{singleflowentry}=96$ bytes.}:
\begin{equation} \label{eq_replymsg}
  l_{reply}(n) = l_{replyheader}+n*l_{singleflowentry}
\end{equation}

\subsubsection{Single Controller}
Given the network graph $G$, the switch set $S$ and the active flow set $F$, let $w_i$ denote the cost of polling all flow statistics from switch $v_i$, $w_i'$ denote the cost of polling a single flow statistics from switch $v_i$. The costs $w_i$ and $w_i'$ are given by:
\begin{enumerate}
\item
  Out-of-band deployment.
  \begin{equation} \label{eq_outofbandweight}
    \begin{array}{ll}
      w_i = l_{req}+l_{reply}(|s_i|), & \forall i \in S \\
      w_i' = l_{req}+l_{reply}(1), & \forall i \in S
    \end{array}
  \end{equation}
\item
  In-band deployment. Given the controller location and the control message routing algorithm, let $h_i$ represent the hops from switch $v_i$ to the controller.
  \begin{equation} \label{eq_inbandweight}
    \begin{array}{ll}
      w_i = (l_{req}+l_{reply}(|s_i|))*h_i, & \forall i \in S \\
      w_i' = (l_{req}+l_{reply}(1))*h_i, & \forall i \in S
    \end{array}
  \end{equation}
\end{enumerate}

Let $q_i$ denote the minimum cost of polling a single flow $i$:
\begin{equation}
  q_i=\min\limits_{i \in s_j}\{w_j'\}, \forall i \in F
\end{equation}

The binary variable $x_i$ indicates whether to poll all flow statistics from switch $v_i$ or not, $y_i$ indicates whether flow $i$ is polled or not. The integer linear programming (ILP) formulation of the problem is given by:
 \begin{equation} \label{eq_problem}
  \begin{split}
    \min & \sum\limits_{i \in S}w_ix_i+ \sum\limits_{i \in F}q_iy_i \\
    \text{subject to:} & \sum\limits_{i \in s_j}x_j+y_i \ge 1, \forall i \in F\\
    & x_i \in \{0,1\}, \forall i \in S \\
    & y_i \in \{0,1\}, \forall i \in F
  \end{split}
\end{equation}

The polling scheme consists of many ``polling all'' and ``polling single'' rules, where the former rules are associated with switches, the latter rules are associated with a mapping $flowid \mapsto switch$. We justify how to obtain the polling scheme from Eq.~\ref{eq_problem}. For all $x_i=1$, adding requests of polling all flow statistics from switch $v_i$; for all $y_i=1$, adding requests of polling single flow statistics from switch $v_{\argmin\limits_{i}{q_i}}$.

\subsubsection{Multiple Controllers} \label{sec_multi}
A more general case is that there are multiple controllers in the network. Any controller can be selected to collect flow statistics from switches. Obviously, selecting a ``nearby'' controller is more cost-efficient. To formulate this problem, we introduce a set $C=\{c_1,c_2,\ldots,c_t\}$ to denote the available controllers. The number of the controllers is $t=|C|$. Let $w_{ij}$ denote the cost of polling all flow statistics in switch $v_i$ from controller $c_j$, $w_{ij}'$ denote the cost of polling a single flow statistics in switch $v_i$ from controller $c_j$. $w_{ij}$ and $w_{ij}'$ are given by:
\begin{enumerate}
\item
  Out-of-band deployment.
  \begin{equation} \label{eq_multioutofbandweight}
    \begin{array}{ll}
      w_{ij} = l_{req}+l_{reply}(|s_i|), & \forall i \in S \\
      w_{ij}' = l_{req}+l_{reply}(1), & \forall i \in S
    \end{array}
  \end{equation}
\item
  In-band deployment. Given the controller locations and the control message routing algorithm, let $h_{ij}$ represent the hops from switch $v_i$ to controller $c_j$.
  \begin{equation} \label{eq_multiinbandweight}
    \begin{array}{ll}
      w_{ij} = (l_{req}+l_{reply}(|s_i|))*h_{ij}, & \forall i \in S, \forall j \in C \\
      w_{ij}' = (l_{req}+l_{reply}(1))*h_{ij}, & \forall i \in S, \forall j \in C
    \end{array}
  \end{equation}
\end{enumerate}

Let $q_{ij}$ denote the minimum cost of assigning controller $c_j$ to poll a single flow $i$:
\begin{equation}
  q_{ij}=\min\limits_{i \in s_k}\{w_{kj}'\}, \forall i \in F, \forall j \in C
\end{equation}

The binary variable $x_{ij}$ indicates whether to poll all flow statistics in switch $v_i$ from controller $c_j$ or not, $y_{ij}$ indicates whether to poll single flow $i$ by controller $c_j$ or not. The integer linear programming (ILP) formulation of the problem is:
 \begin{equation} \label{eq_multiproblem}
  \begin{split}
    \min & \sum\limits_{i \in S}\sum\limits_{j \in C}w_{ij}x_{ij}+ \sum\limits_{i \in F}\sum\limits_{j \in C}q_{ij}y_{ij} \\
    \text{subject to:} & \sum\limits_{i \in s_j}\sum\limits_{k \in C}x_{jk}+\sum\limits_{j \in C}y_{ij} \ge 1, \forall i \in F\\
    & x_{ij} \in \{0,1\}, \forall i \in S, \forall j \in C \\
    & y_{ij} \in \{0,1\}, \forall i \in F, \forall j \in C
  \end{split}
\end{equation}

We justify how to obtain the polling scheme from Eq.~\ref{eq_multiproblem}. For all $x_{ij}=1$, adding requests of polling all flow statistics from switch $v_i$ by controller $c_j$; for all $y_{ij}=1$, adding requests of polling single flow statistics from switch $v_{\argmin\limits_{i}{q_{ij}}}$ by controller $c_j$.
\begin{algorithm}[!t] \footnotesize
  \KwIn{$G=(V,E)$: the network; $F$: the active flows; $H$: the number of hops vector}
  \KwOut{$S$: candidate polling set, $W$ the corresponding cost vector}
  $S \gets \{(s_1: \emptyset),(s_2: \emptyset),\ldots, (s_n: \emptyset)\}$ \;
  $W \gets []$ \tcp*{the weight vector for $S$}
  \ForEach{$f \in F$} {
    \ForEach{$v \in P_f$} {
      $S[v] \gets S[v].append(f)$\;
    }
    $S \gets S \cup \{f\}$ \tcp*{Add single flow polling set}
  }
  \ForEach{$s \in S$} {
    $W[s]=(l_{req}+l_{reply}(|s|))*H[s]$
  }
  \Return{$S, W$}
  \caption{Construct Cost Functions}
  \label{algo_constructsetcover}
\end{algorithm}

\subsection{Solutions}
In this section, we first describe the optimal solution for this problem by exhaustive search. As the computing complexity of the aforementioned formulations are NP-hard, we propose heuristics to approximate the optimal performance.

\subsubsection{Optimal Solution}
The optimal solution is the minimum cost among the sum of all possible combinations of sets, which covers all the active flows. It can be obtained by a brute-force search algorithm, which is shown in Algorithm~\ref{algo_brute}. We refer to this algorithm as ``optimal''. The size of the given sets is $m+n$, where $m$ is the active flow number and $n$ is the switch number. The size of the possible combinations is $\sum\limits_{l=1}^{|S|}{|S|\choose l}=2^l-1$. Therefore, the complexity of the brute-force search algorithm is $O(2^{m+n})$, which is exponential to the number of switches and active flows. This optimal algorithm is not scalable for large networks with plenty of active flows.

\subsubsection{Heuristics}
Both formulations Eq.~\ref{eq_problem} and Eq.~\ref{eq_multiproblem} are the weighted set cover problem, which is proved to be NP-hard~\cite{approximationbook}. We propose a greedy strategy which selects the most cost-effective switches until all the active flows are covered. The algorithm is shown in Algorithm~\ref{algo_greedy}. The main loop iterates for $O(n)$ time, where $n=|F|$. The most cost-effective set $s$ can be found in $O(\log m)$ time by a priority queue, where $m=|S|$. So the computational complexity of the algorithm is $O(n \log m)$. The analysis of the algorithm is listed below. Without loss of generality, we define the sets as the union of the polling all and the polling single sets. Each set is represented as $P_i$, which is a set of polling flows.
\begin{theorem}
  Algorithm~\ref{algo_greedy} is an $H(p)$-approximation, where $p=\max_i\{|P_i|\}$, $H(p)$ is the $p$-th harmonic number.
\end{theorem}

\begin{proof}
  Assume that the algorithm selects the polling set $P_1, P_2, \ldots P_k$ in this order to form the polling scheme. Consider a flow $f$ which is first covered when $P_i$ is selected. Suppose $R$ is the set of remaining uncovered flows when $P_i$ is selected. Define the cost of covering a flow $f$ as $c_f = \frac{w(P_i)}{|P_i \cap R|}$, where $w(P_i)$ is the cost of polling $P_i$. According to the definition of $c_f$, the following equality holds:
  \begin{equation} \label{eq_setcover1}
    \sum\limits_{P_i \in C}w(P_i) = \sum\limits_{f \in F}c_f
  \end{equation}
  For an arbitrary polling set $P_k=\{f_1, f_2, \ldots f_d\}$, suppose $f_i$ is selected before $f_j$ if $i \le j$. When $f_j$ is covered, $R \subseteq \{f_j, f_{j+1} \ldots, f_d\}$. Therefore, the cost of polling $P_k$ is $\frac{w(P_k)}{|P_k \cap R|} \le \frac{w(P_k)}{d - j + 1}$. Assume $P_i$ is the selected set by the algorithm. Since $P_i$ is the most cost-efficient polling set, we have:
  \begin{equation} 
    \frac{w(P_i)}{|P_i \cap R|} \le \frac{w(P_k)}{|P_k \cap R|} \le \frac{w(P_k)}{d - j + 1}
  \end{equation}
  Then, the sum of the cost of all elements in $P_k$ is given by:
  \begin{equation} 
    \sum\limits_{f \in P_k}{c_f}=\sum\limits_{i=1}^d{\frac{w(P_i)}{|P_i \cap R|}} \le \sum\limits_{i=1}^d{\frac{w(P_k)}{d-i+1}} = H(d)w(P_k)
  \end{equation}
  Consider $p=\max_i\{|P_i|\}$ and let $P_i$ denote the corresponding polling set, we have:
  \begin{equation} \label{eq_setcover2}
    H(p)w(P_i) \ge \sum\limits_{f \in P_i}c_f
  \end{equation}
  Since the total number of a set cover elements is greater or equal to the number of elements in the universe $F$, the inequality holds:
  \begin{equation} \label{eq_setcover3}
    \sum\limits_{P_i \in C^*}\sum\limits_{f \in P_i}c_f \ge \sum\limits_{f \in F}c_f
  \end{equation}
  Let $C^*$ denote the optimal polling scheme, $C$ denote the polling scheme generated by our algorithm. Combining Eq.~\ref{eq_setcover1}, Eq.~\ref{eq_setcover2} and Eq.~\ref{eq_setcover3}:
  \begin{equation}
    \begin{split}
      w(C^*) &=\sum\limits_{P_i \in C^*}w(P_i) \ge \sum\limits_{P_i \in C^*}{\frac{1}{H(p)}\sum\limits_{f \in P_i}c_f} \\
      &=\frac{1}{H(p)}\sum\limits_{P_i \in C^*}\sum\limits_{f \in P_i}c_f \\
      &\ge \frac{1}{H(p)}\sum\limits_{f \in F}c_f = \frac{1}{H(p)}\sum\limits_{P_i \in C}w(P_i)\\
      &=\frac{1}{H(p)}w(C)
    \end{split}
  \end{equation}
  This shows that Algorithm~\ref{algo_greedy} is an $H(p)$-approximation.
\end{proof}

Next, we analyze the accuracy of the heuristic. Due to the congestion link and the matching issues, the collected flow counter may be different from the real one. We emulate the packet loss by introducing two parameters: packet loss rate $r$ and loss switch ratio $p$. The switches are divided into two categories: normal switch and loss switch. When a packet passes a loss switch, it is dropped with a probability of the packet loss rate.
\begin{theorem}
  The accuracy of a flow can be estimated by $1-\frac{(lp+1)r(1-r)^{lp}}{1-(1-r)^{lp+1}}$, where $l$ is the number of switches along its routing path.
\end{theorem}

\begin{proof}
  The selection of a switch along the routing path of the flow can be regarded as a random event (no matter by ``polling all'' or by ``polling single''). The number of loss switches for this flow can be obtained by $l \cdot p$. Let $c^*$ and $c$ denote the real flow counter and the polled flow counter respectively. If there are $i$ congested switches, the real flow counter $c^*$ can be computed by $\frac{c}{(1-r)^i}$. Assume the congested switches are placed randomly along its routing path. The real flow counter can be computed by enumerating all possible number of congested switches along its path:
  \begin{align}
    c^* &= \frac{1}{lp+1}\sum\limits_{i=0}^{lp}\frac{c}{(1-r)^i} \notag\\
    &=\frac{c}{lp+1} \cdot \frac{1-(\frac{1}{1-r})^{lp+1}}{1-\frac{1}{1-r}} \notag\\
    &= \frac{c}{lp+1} \cdot \frac{1-(1-r)^{lp+1}}{r(1-r)^{lp}} \notag
  \end{align}
  Then, the accuracy of the flow can be computed as:
  \begin{displaymath}
    1-\frac{c}{c^*}= 1-\frac{(lp+1)r(1-r)^{lp}}{1-(1-r)^{lp+1}}
  \end{displaymath}
\end{proof}

\begin{algorithm}[!t] \footnotesize
  \KwIn{$S$: candidate polling set; $W$ the corresponding cost vector}
  \KwOut{$Pa$: the polling all set; $Pb$: the polling single set; $mincost$}
  $Pa \gets [], Pb \gets []$ \;
  $mincost \gets +\infty$ \;
  \ForEach{$l \gets 1$ to $|S|$} {
    \While{$C \gets NextCombination(S, l)$} {
      \If{$Cost(C) < mincost$} {
        $mincost \gets Cost(C)$ \;
        $Pa \gets PollAll(C)$ \;
        $Pb \gets PollSingle(C)$ \;
      }
    }
  }
  \Return{$Pa, Pb, mincost$}
  \caption{Optimal Polling Scheme Generation}
  \label{algo_brute}
\end{algorithm}

\subsection{Handling Flow Dynamics}
CeMon generates the optimized polling scheme periodically to keep the scheme updated. However, the active flows in the network change from time to time and make the current polling scheme sub-optimal. In this section, we propose a novel heuristic to handle flow dynamics. We also discuss how to strike a trade-off between the computing efficiency and the performance of the polling scheme by adaptively adjusting the reconstruction frequency.

\subsubsection{Heuristic}
CeMon detects the flow dynamics by the flow state tracker. Intuitively, the polling scheme optimizer has to re-calculate the polling scheme upon receiving flow arrival/expiration messages. However, we argue that this is not necessarily true in practice. So we propose another heuristic called ``Dynamic Adjust and Periodical Reconstruction'' (DAPR) to handle flow dynamics:
\begin{itemize}
\item
  When a new flow arrives: if it has been covered by the current polling scheme, no further actions are needed. Otherwise, just add one single flow polling to the polling scheme.
\item
  When a flow expires: if this flow is collected by a single flow polling, remove it from the polling scheme. Otherwise, no actions.
\end{itemize}

The DAPR cannot always keep the polling scheme optimal, because patching the current polling scheme by adding or removing single polling rules has no performance guarantee. However, DAPR prevents the polling scheme from changing too frequently to impose extra overhead on the controller. To keep the polling scheme up to date, we reconstruct the polling scheme periodically.

\subsubsection{Reconstruction Interval}
The DAPR tries to patch the polling scheme to enable it to tolerate the flow dynamics. However, as time elapses, too many patches make the polling scheme sub-optimal and degrade its performance. The question is when to reconstruct the polling scheme? Obviously, reconstruction with a high frequency yields too much computing overhead, while a low frequency cannot guarantee the scheme performance and keep it updated.
\begin{algorithm}[!t] \footnotesize
  \KwIn{$S$: candidate polling set; $W$ the corresponding cost vector}
  \KwOut{$Pa$: the polling all set; $Pb$: the polling single set; $mincost$}
  $Pa \gets [], Pb \gets []$ \;
  $C \gets \emptyset; mincost \gets +\infty$ \;
  \While{$C \neq U$} {
    Find a set $s \in S$ such that $\frac{W[s]}{|s-C|}$ is minimum \;
    \If{$IsPollingAll(s)$} {
      $Pa.append(s)$ \;
    }
    \Else {
      $Pb.append(s)$
    }
    $mincost += W[s]$
  }
  \Return{$Pa, Pb, mincost$}
  \caption{Polling Scheme Generation Heuristic}
  \label{algo_greedy}
\end{algorithm}

A straightforward method is setting a fixed interval to reconstruct the polling scheme. However, the drawback is that it is not responsive to the dramatic flow change. Therefore, we propose an adaptive reconstruction interval (ARI) which takes the flow variance rate into account. Assume $F_r$ is the corresponding active flow set for the latest reconstructed polling scheme, $F_c$ is the current active flow set. Define the flow variance rate $D$ as:
\begin{equation}
  D(F_r, F_c)=\frac{|F_r\cap F_c|}{|F_r|}
\end{equation}

The flow variance rate shows how many flows are still in the original flow sets: the smaller the value, the more flows change in the network. A threshold $\tau$ is provided to measure the degree of the flow variance rate: when $D(F_r, F_c) < \tau$, the polling scheme will be reconstructed. We evaluate the performance of ARI in Section~\ref{sec_evaluation}.

\section{Adaptive Fine-grained Polling Scheme} \label{sec_afps}
The maximum coverage polling scheme collects statistics of all active flows from switches by aggregating polling requests and replies. On the other hand, in many real-world scenarios, measurement applications usually associate with a subset of flows. Therefore, a light-weight and fine-grained flow level polling scheme are necessary for measurement applications. In this section, we propose Adaptive Fine-grained Polling Scheme (AFPS) as a complementary scheme for MCPS. We first formulate the flow level measurements. Then, adaptive sampling algorithms are developed to deliver timely flow information without incurring too much polling overhead.

\subsection{Overview}
Software defined measurements are usually conducted on top of flows. Therefore, flow level monitoring is a fine-grained measurement implementation for upper layer applications. Since the current SDN architecture does not have complex functionalities in the switch, active polling is a practical solution to collect flow statistics. In essence, the active polling is a sort of sampling. The controller gathers the flow statistics periodically by querying switches, and computes the difference between the last two readings. However, the sampling frequency is critical for such measurement implementations. Low sampling frequency imposes less overhead, but it has a high probability of tracking instant traffic changes. To the contrary, high sampling frequency is more likely to identify the traffic spikes, but the communication overhead cannot be neglected. CeMon strikes a better trade-off between the measurement accuracy and overhead by dynamically tuning the sampling frequency in terms of measured traffic. The basic idea is that we collect the flow statistics more frequently when the traffic is busy, and decrease the sampling rate when there is less traffic.

\subsection{Problem Formulation}
Flow level measurements can be formulated as follows. A task $T$ is associated with a flow set $F=\{f_1,f_2, \ldots,f_n\}$. For each flow $f_i \in F$, we poll flow statistics in time $\{s_1,$ $s_2, \ldots, s_m\}$, and obtain the corresponding reading: $C_{f_i}=\{c_{f_1}(s_1),$ $c_{f_2}(s_2), \ldots, c_{f_m}(s_m)\}$. The monitoring result at time $t$ can be defined as a function $M_T(C_{f_1},C_{f_2}, \ldots, C_{f_n}, t)$, where $M_T$ is the operation function for $T$ on the current readings. For example, a link utilization task should measure the link usage during a period $[t-\tau, t)$. Define the operation function $M_T$ as sum, the instant utilization at $t$ is summing all the active flows' utilization during $[t-\tau,t)$. Notice that the corresponding flows have different sampling rates, the current counter can only be obtained by the latest polling of the flow: $c(t)=\mathop{c}\limits_{\argmax\limits_{i}{s_i \le t}}(s_i)$. Therefore, the link utilization task is formulated by:
\begin{equation} \label{eq_lu}
  U(t, t - \tau) = \sum\limits_{i=1}^{n}[\mathop{c_{f_i}}_{\argmax\limits_{i}{s_i \le t}}(s_i) - \mathop{c_{f_i}}_{\argmax\limits_{j}{s_j \le t-\tau}}(s_j)]
\end{equation}

The above example illustrates the link utilization task, however, our formulation can easily be extended to various monitoring tasks. Define the operation function as max, we can detect the heavy hitter flows~\cite{datasummary} in the network; define $M_T$ as a function that sums the flows with the same source IP and destination IP, the formulation describes the traffic matrix estimation. These examples do not intend to show the tricks to formulate the monitoring applications, but to demonstrate that our framework is generic enough to support a wide variety of tasks at a flow level measurement granularity.

\subsection{Tuning Sampling Frequency}
Timely flow statistics collection is crucial for many measurement tasks. The key challenge is how to determine the sampling rate at low-cost and high-accuracy. A straightforward approach is polling flow statistics at a fixed rate, we refer to this method as ``fixed sampling''. The drawback is that it wastes resources when the traffic is slow and cannot grab the traffic spikes in a timely fashion. As a result, adaptive sampling algorithms which adjust the polling frequency according to traffic dynamics are needed. The sampling frequency tuning algorithms should be light-weight, memory-efficient and responsive to traffic changes. To avoid excessive polling, a valid sampling frequency range is provided for all algorithms: $[\tau_{min}, \tau_{max}]$. Also, if a flow is expired before the next polling, CeMon can obtain the flow statistics by its \texttt{FlowRemoved} message. We detail the proposed algorithms in the following sections.

\subsubsection{Proportional Tuning}
To dynamically adjust the sampling frequency, we predict the future packet arrival rates based on historical data. Specifically, a straightforward approach is tuning the sampling frequency according to the traffic change rate: the more the traffic varies, the less the sampling interval and vice versa. Let $\tau_n$ represent the interval at the $n$th sampling, it can be derived from $\tau_{n-1}$:
\begin{equation}
  \tau_n^{pt} = s_{n+1} - s_n = \tau_{n-1} \cdot v \cdot \frac{s_n-s_{n-1}}{c(s_n)-c(s_{n-1})}
\end{equation}

where $v$ is a coefficient for the average of the current traffic volume. We refer to this algorithm as ``Proportional Tuning'' (PT), because the sampling frequency is in proportion to the traffic change rate.

\subsubsection{EWMA Tuning}
PT works well when the traffic is relatively stable. However, the sampling interval generated by PT may fluctuate when the traffic changes dramatically. To avoid such fluctuations, we improve PT by employing a smoothing technique named Exponentially Weighted Moving Average (EWMA)~\cite{ewma}. EWMA takes more historical data into account while placing more emphasis on recent data. The $n$th sampling interval $\tau_n$ is given by:
\begin{equation}
  \tau_n^{ewma} = \alpha \cdot \tau_n^{pt}+(1-\alpha) \cdot \tau_{n-1}^{ewma}
\end{equation}

Where $\alpha$ is a constant smoothing factor between $[0,1]$. We refer to this algorithm as ``EWMA Tuning'' (EWMAT).

\subsubsection{Sliding Window Based Tuning}
Previous tuning algorithms require parameters such as the traffic factor and the smoothing factor. In practice, parameter determination is difficult and error-prone. As such, we develop a Sliding Window based Tuning algorithm (SWT), which adjusts the sampling frequency regarding the statistical measures in a parameter-free style.
\begin{algorithm}[!t] \footnotesize
  \KwIn{$f$: the target flow}
  \KwOut{Register the next polling time}
  $win \gets []$ \tcp*{the sliding window deque}
  $ws \gets 3$ \tcp*{the initial window size}
  $var = Poll(f) - lastreading$ \;
  \If{$var > win.mean + 2*win.stdev$} {
    \tcp{The traffic changes significantly}
    $\tau \gets \max(\tau_{min}, \tau / 2)$ \;
    $ws \gets \min(3, \lceil ws / 2 \rceil)$
  }
  \Else {
    $\tau = \min(\tau_{max}, \tau * 2)$ \;
    $ws \gets ws + 1$
  }
  \If{$win.length > ws$} {
    $win.popfront()$
  }
  $PollEventhandler.Add(f, \tau)$ \;
  \Return{}
  \caption{Sliding Window Based Tuning}
  \label{algo_slidewindow}
\end{algorithm}

A pseudo-code description of SWT for a flow is depicted in Algorithm~\ref{algo_slidewindow}. We maintain a sliding window to store the recent transmitted bytes for the flow. Each time after reading counter in the switch, we judge whether the latest traffic is significantly different from the traffic in the sliding window. If so, we decrease the sampling interval by half. Otherwise, double the sampling interval and update the sliding window with the latest data. Additive-Increase/Multiplicative-Decrease (AIMD) paradigm ~\cite{aimd} is also employed to adjust the window size. The rationale behind this is that when the traffic does not change a lot, the window size should be expanded to keep the recent data stable. Otherwise, the windows size should be decreased quickly to be responsive to instant traffic spikes.

In order to evaluate the performance of the proposed algorithms, we define the measurement error as:
\begin{equation} \label{eq_error}
  R=\frac{1}{N}\sqrt{\sum\limits_{i=1}^{N}{(x_i-\hat{x_i})^2}}
\end{equation}

Where $x_i$ and $\hat{x_i}$ are the actual and measured flow statistics at the $i$th sample respectively, $N$ is the number of polling samples. For link utilization tasks, $x_i=c(s_i)$.
\begin{figure}[!t]
  \centering
  \includegraphics[width=\linewidth]{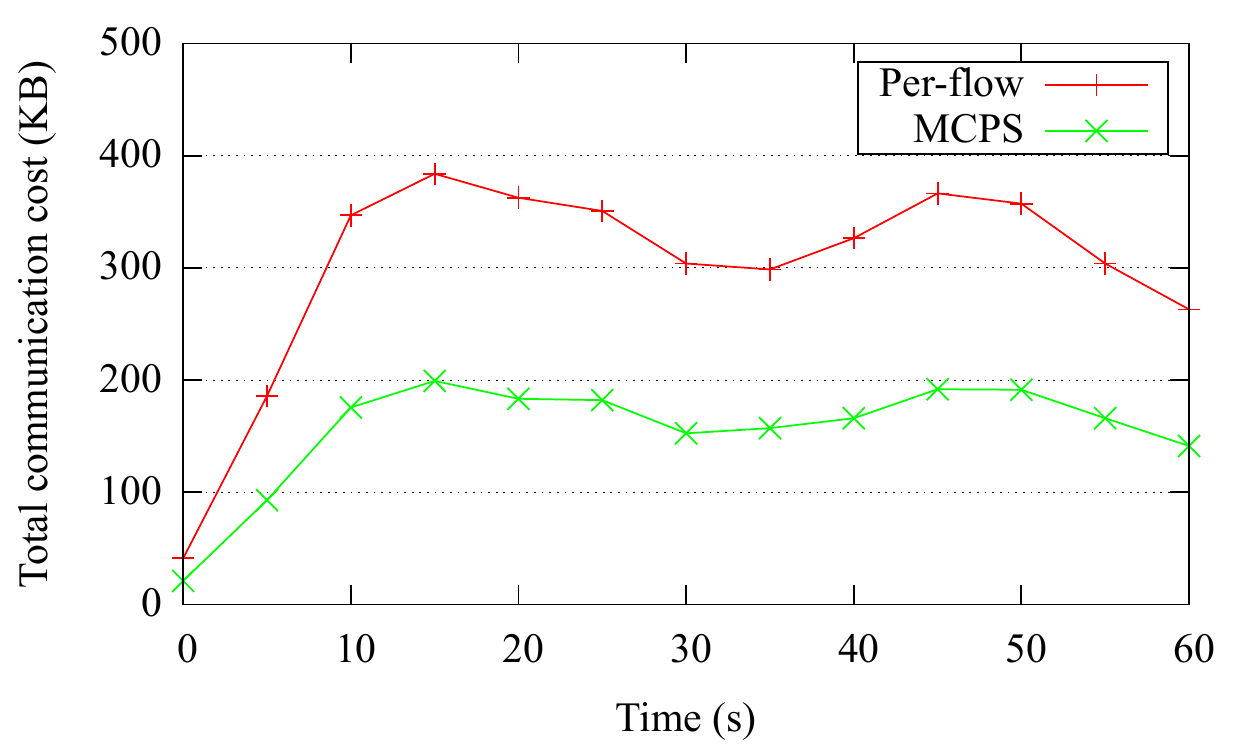}
  \caption{The communication cost comparison of Abilene topology.}
  \label{fig_emulate}
\end{figure}

\section{Evaluation} \label{sec_evaluation}
We evaluate the performance of the MCPS and the AFPS from different aspects such as the reduced communication cost, overhead, accuracy and handling flow dynamics. Extensive experimental results demonstrate that CeMon significantly reduces the monitoring cost with negligible loss of accuracy. 

\subsection{Evaluation Methodology}
\textbf{Network topology and flow generation.} In our prototype emulation, we use a real network topology ``Abilene'' from the Internet topology zoo~\cite{topologyzoo} which consists of $10$ switches. In large scale simulation, two widely used network graph models Erd\H{o}s-R\'enyi graph~\cite{erdos} and Waxman graph~\cite{waxman} are applied to generate huge network graphs to demonstrate the efficiency of our schemes. Unless specified, the experiments are conducted in an Erd\H{o}s-R\'enyi graph with $200$ switches. We generate flows and choose the source and destination from all the hosts in a uniformly random manner.

\textbf{Prototype implementation.}
We implement a prototype of CeMon as a module of POX controller~\cite{pox} to verify its feasibility. We emulate the Abilene network by Mininet~\cite{mininet}, which is a famous emulator in SDN. The experiments are conducted on software switches~\cite{openvswitch}. We use the shortest path algorithm to generate the routing path for each flow. The flows are generated according to a $60$s packet trace UNI1 collected from a datacenter~\cite{trafficchar}.

\textbf{Trace-driven simulation.}
Since the emulation of large networks are resource-hungry and infeasible, we conduct large scale experiments by building a trace-driven simulator written in Python. For the experiments of the DAPR and the AFPS, we use real packet traces collected from a data center~\cite{trafficchar} to perform the simulation. Since the simulation only cares about the active flows and their forwarding paths, we need not replay the traces. Instead, we only simulate the event of flow arrival and expiration in the network. For the DAPR, a $60$s packet trace UNI1 is employed. For the AFPS, two $60$s packet traces UNI1 and UNI2 are employed to represent TCP and UDP traffic respectively.

\textbf{Experiment setup.}
All experiments are conducted on a server equipped with an Intel i7-4770 3.40 GHz CPU processor and 32G RAM. The server runs Ubuntu~12.04 operating system with Python~2.7.

\begin{figure}[!t]
  \centering
  \includegraphics[width=\linewidth]{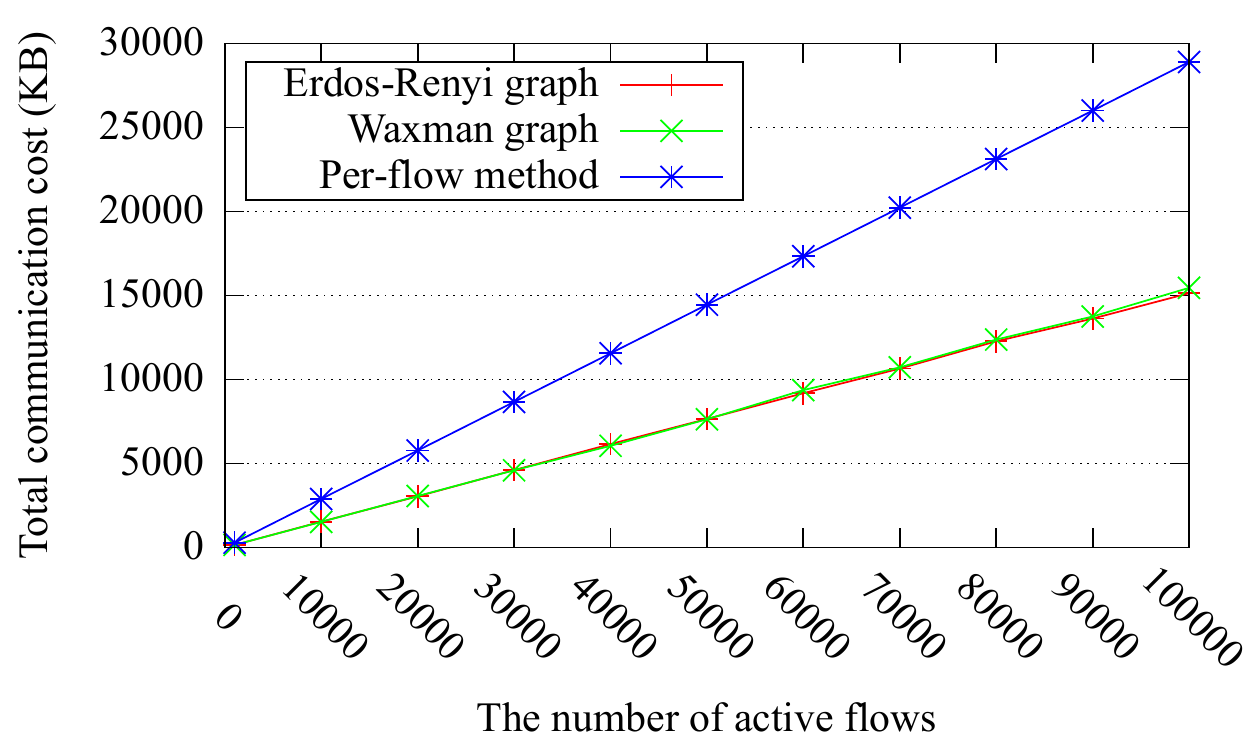}
  \caption{Communication cost for different graph models (out-of-band).}
  \label{fig_cost}
\end{figure}

\subsection{MCPS Results}
\subsubsection{Communication Cost}
We first demonstrate the effectiveness of MCPS in terms of communication cost. We compare it with the basic per-flow querying method proposed in~\cite{opentm}. Figure~\ref{fig_emulate} shows the communication cost from our prototype. The peak number of active flows is $1297$ and the polling interval is $5$s. Obviously, the MCPS is consistently superior to the per-flow querying in terms of the communication cost. Specifically, for the total twelve pollings, the MCPS saves $48.1\%$ of the total monitoring cost on average.

Large scale experiments are conducted by the simulator. Figure~\ref{fig_cost} shows the total communication cost in Erd\H{o}s-R\'enyi graph and Waxman graph for out-of-band deployment. The number of active flows varies from $1000$ to $100000$ which is huge enough for a medium-sized data center. The total communication cost of the per-flow polling method is irrelevant to the network topology, but in proportion to the number of flows. As a result, we plot only one curve for reference in Figure~\ref{fig_cost}. Compared with the per-flow polling method, MCPS significantly reduces the communication cost in both network topologies. It saves up to $47.6\%$ of the total communication cost. The result conforms to our emulation in Figure~\ref{fig_emulate}. Figure~\ref{fig_inband} investigates the effectiveness of MCPS for the in-band deployment, where the ``random per-flow'' strategy is querying the switch randomly along the routing path for each flow, the ``minimum per-flow'' strategy is querying the switch that consumes minimum bandwidth. Clearly, MCPS consistently outperforms the per-flow querying strategy by reducing roughly $50\%$ of the cost as the number of active flows varies. Compared with the original random per-flow query, MCPS saves the cost by up to $50.2\%$ which is slightly better than out-of-band deployment scenario.
\begin{figure}[!t]
  \centering
  \includegraphics[width=\linewidth]{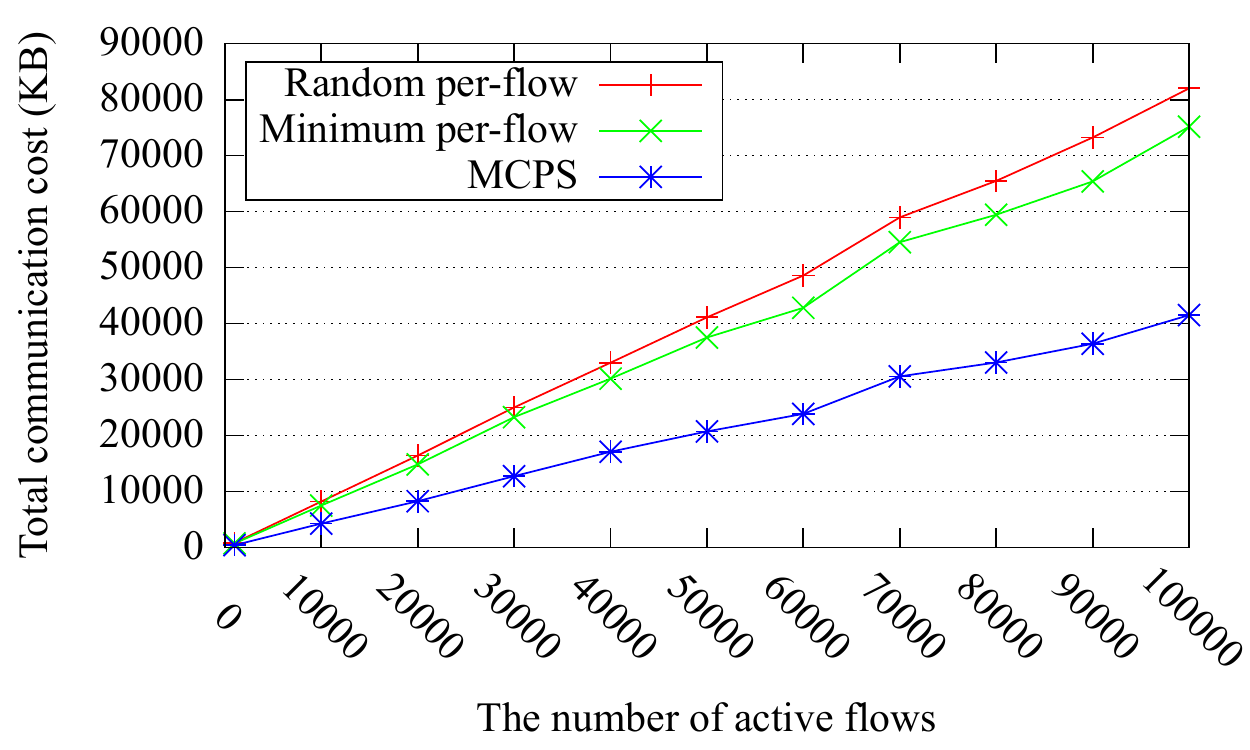}
  \caption{Communication cost for in-band deployment.}
  \label{fig_inband}
\end{figure}

As stated in Section~\ref{sec_multi}, MCPS can be extended to networks with multiple controllers. Since any controller is able to monitor any flow, the ``random per-flow'' strategy is querying a switch by a randomly chosen controller, the ``minimum per-flow'' strategy is querying a switch by the controller which incurs minimum cost. Consider there are more available controllers, the minimum cost of querying a flow is decreased since the average hops from a flow to a controller is shortened. As a result, MCPS can further reduce the communication cost in this case. Figure~\ref{fig_multi} shows the communication cost in a network with $3$ controllers. It shows that MCPS saves $57\%$ of the cost on average as the number of active flow increases. Furthermore, Figure~\ref{fig_controllernum} shows the communication cost gradually decreases when the number of controllers increases. We obtain about $2\%$ to $3\%$ of the cost reduction by adding one controller. More than $10\%$ of the cost reduction can be gained by increasing the controller number from $1$ to $5$.

These experiments illustrate the effectiveness and generality of MCPS: it consistently reduces the communication cost by roughly $50\%$ for polling all flow statistics, regardless of the number of active flows, network topologies, deployment methods, and the number of controllers.
\begin{figure}[!t]
  \centering
  \includegraphics[width=\linewidth]{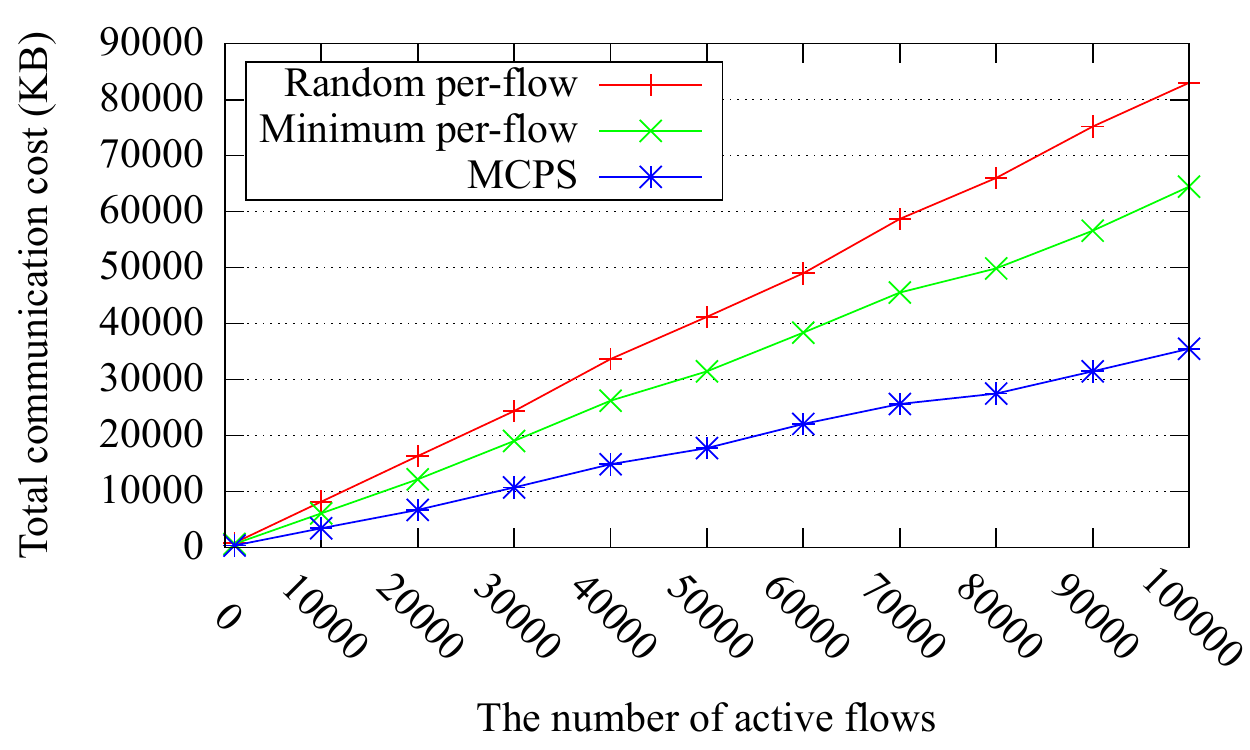}
  \caption{Communication cost for multiple controllers.}
  \label{fig_multi}
\end{figure}

\subsubsection{Efficiency}
In this section, we compare the running time and the gap between the optimal solution and the proposed algorithm for MCPS. Figure~\ref{fig_optvsflow} and Figure~\ref{fig_optvsswitch} show that the greedy heuristic performs fairly close to the optimal as the number of active flows and switches increases. Specifically, the maximum difference between the optimal and our heuristic is less than $7\%$. However, the optimal running time is pretty long. Even for a small network with $10$ switches and $20$ active flows, it generates the optimal solution for more than $5000$s. This is impractical for real time network monitoring as we need to poll the flow statistics at second level~\cite{hedera}. The optimal running time increases exponentially as the number of active flows and switches increases. Comparatively, our heuristic is practical and efficient as it produces near-optimal results in less than $1$ms which is 1 million times faster.
\begin{figure}[!t]
  \centering
  \includegraphics[width=\linewidth]{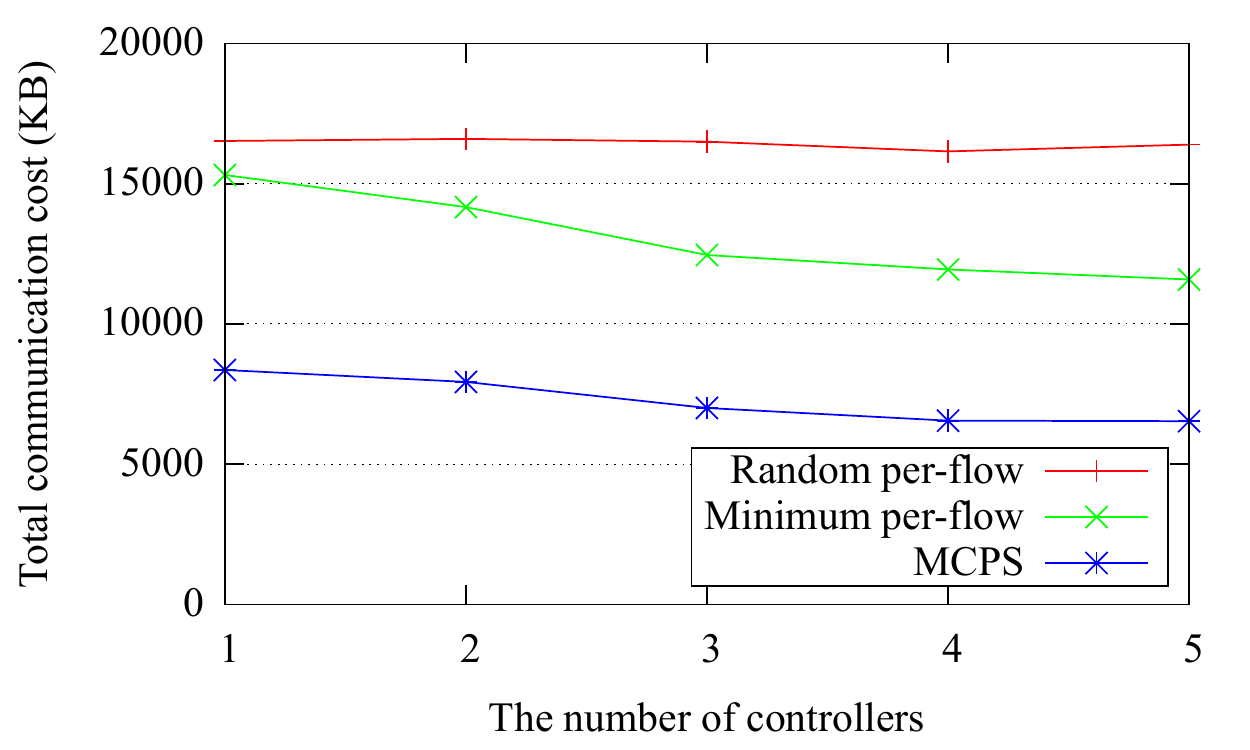}
  \caption{The communication cost as the number of controllers varies (with 20000 active flows in the network).}
  \label{fig_controllernum}
\end{figure}

We examine the construction time of the weighted set cover and the polling scheme generation time in Figure~\ref{fig_constructsetcover}. There is a steady increase in the total computing time over the number of active flows. The problem construction time occupies roughly $10\%$ of the total calculation time. The polling scheme computing time is proportional to the number of active flows (with fixed number of switches) which conforms to the complexity of the greedy algorithm. Our approach obtains the optimized polling scheme very efficiently in practice: for a network with up to $100000$ active flows, we get the polling scheme within $1.6$s. The computation time is able to meet real-world monitoring application polling frequency, such as Hedera~\cite{hedera} which is a data center dynamic flow scheduling system with a control loop of five seconds).
\begin{figure}[!t]
  \centering
  \includegraphics[width=\linewidth]{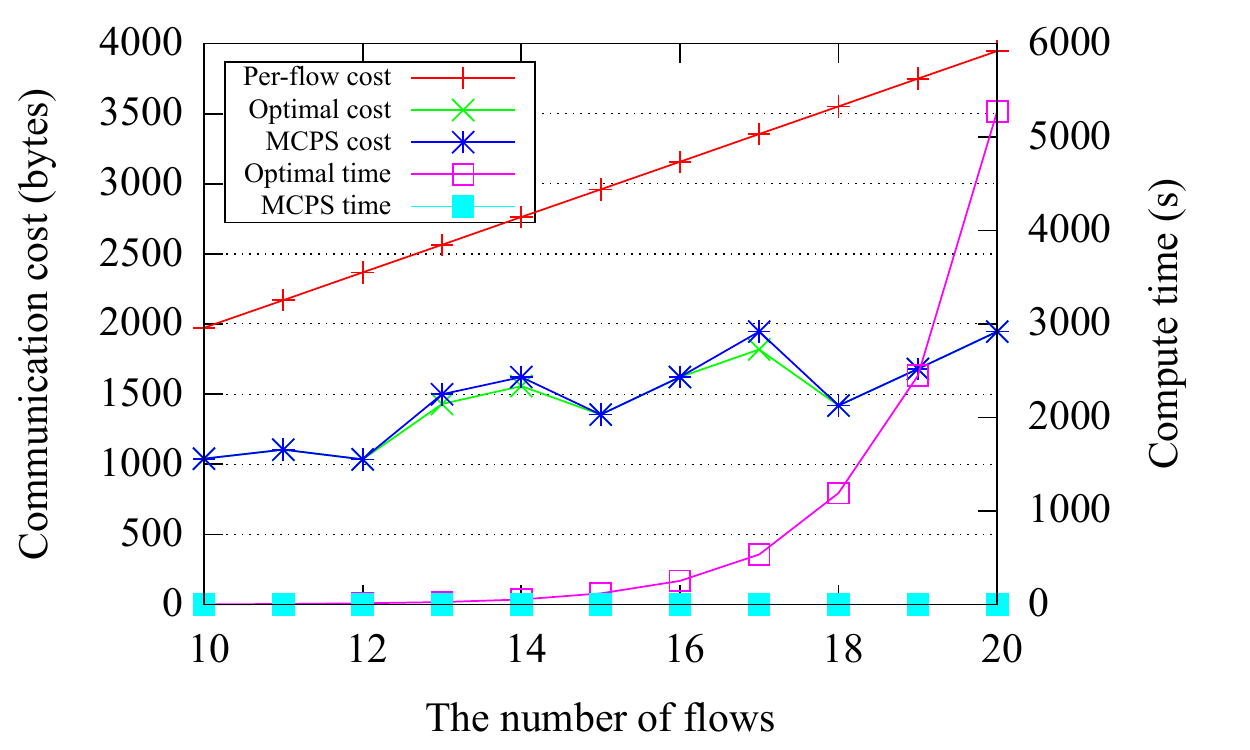}
  \caption{Comparison of the optimal solution and the heuristic as the number of active flows varies.}
  \label{fig_optvsflow}
\end{figure}

\begin{figure}[!t]
  \centering
  \includegraphics[width=\linewidth]{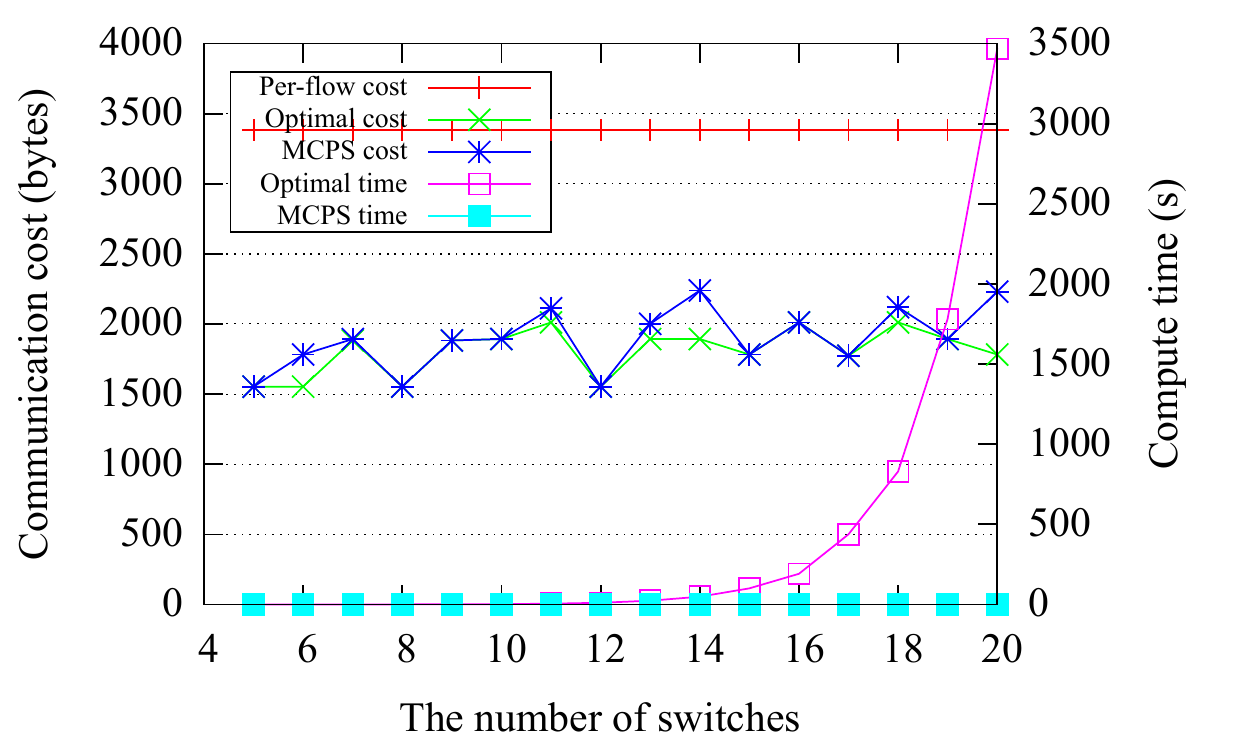}
  \caption{Comparison of the optimal solution and the heuristic as the number of switches varies.}
  \label{fig_optvsswitch}
\end{figure}

The relation between the number of switches and the polling scheme computing time is explored as well. As shown in Figure~\ref{fig_caltimevsswitch}, for $20000$ active flows in a Erd\H{o}s-R\'enyi graph, the computing time for the polling scheme keeps relatively stable, since the computing time is in logarithm relation with the number of switches.

\subsubsection{Accuracy}
We evaluate the accuracy of MCPS by two metrics: accurate flow ratio (AFR) which is obtained by the number of accurate measured flows over the total number of flows; average accuracy of traffic matrix (TM) estimation which is obtained by accurate measured matrix elements over the total number of the elements in the traffic matrix. Loss switch ratio is defined as the number of loss switches to the number of all switches. We generate loss switches in a uniformly random manner according to the loss switch ratio.
\begin{figure}[!t]
  \centering
  \includegraphics[width=\linewidth]{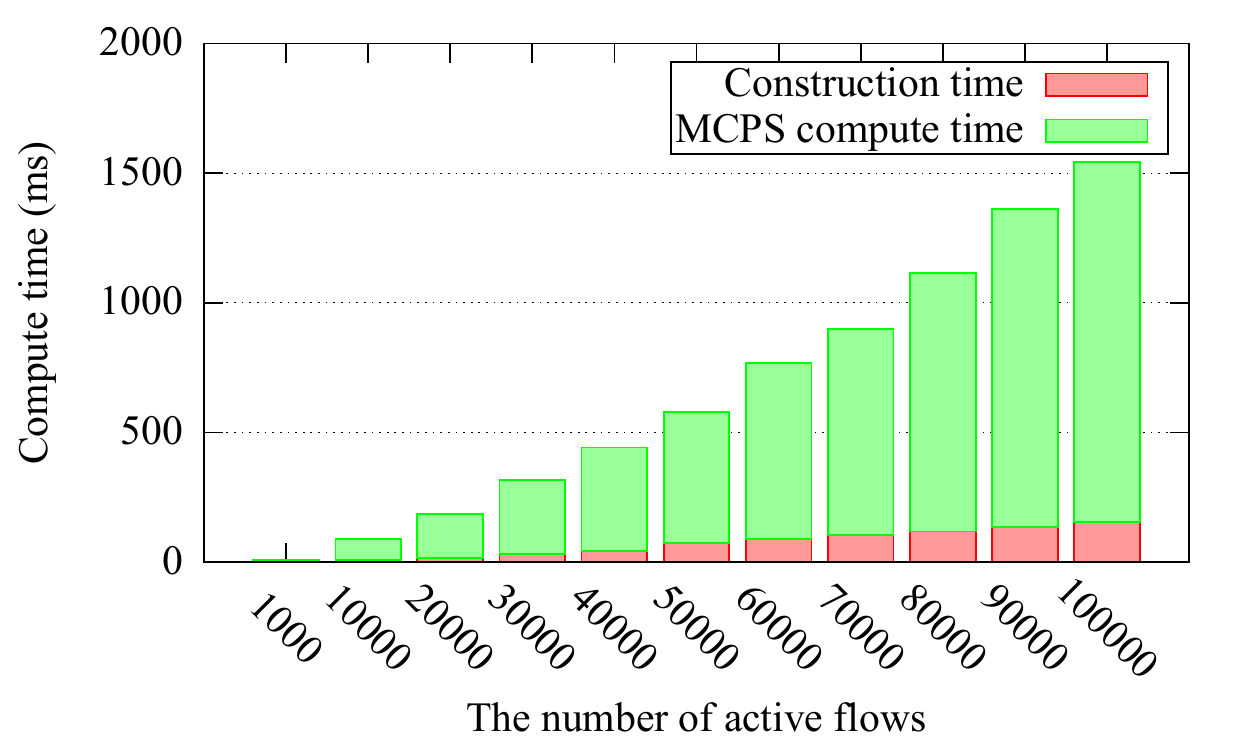}
  \caption{The weighted set cover construction time and the solution computing time.}
  \label{fig_constructsetcover}
\end{figure}
  
\begin{figure}[!t]
  \centering
  \includegraphics[width=\linewidth]{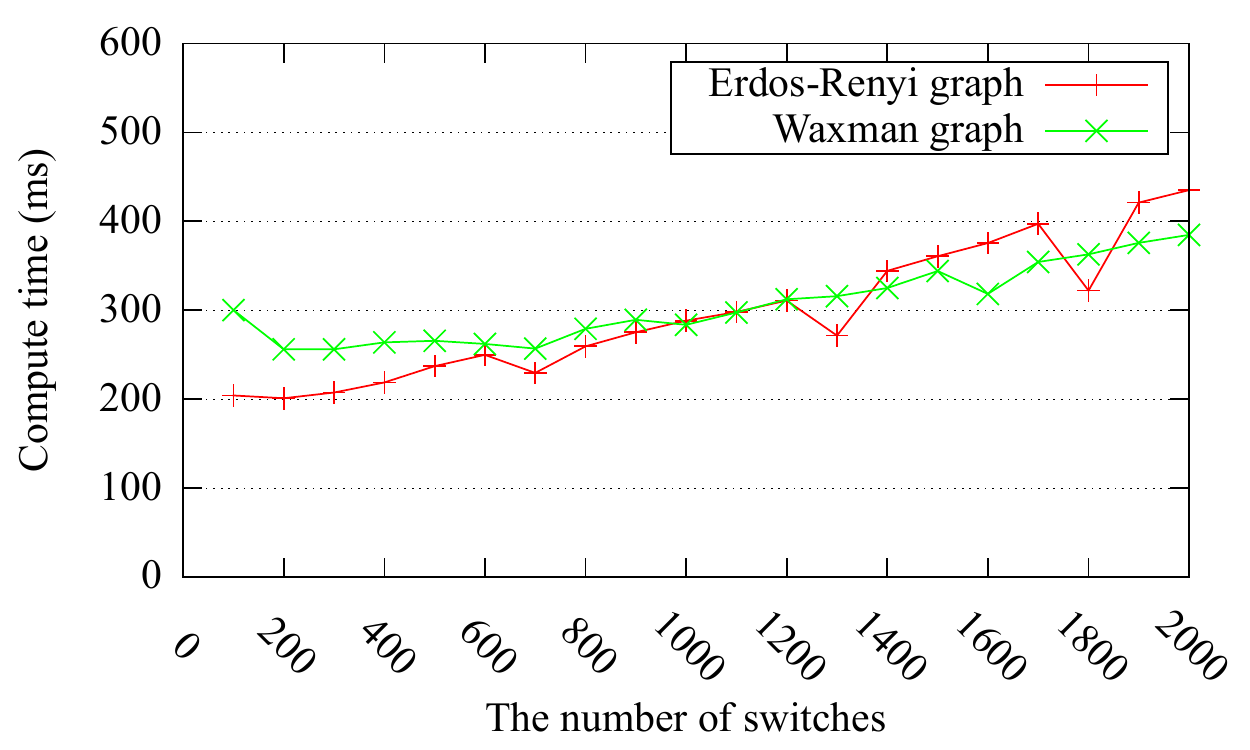}
  \caption{The total computing time vs. the number of switches ($20000$ active flows).}
  \label{fig_caltimevsswitch}
\end{figure}

Figure~\ref{fig_accuracyvspacketloss} illustrates that the AFR is robust to the increasing packet loss rate; the accuracy of TM estimation falls gradually from $99.9\%$ to $98.1\%$. Figure~\ref{fig_accuracyvslossswitch} shows that the AFR falls in proportion to the loss switch ratio. However, the accuracy of TM estimation only decreases slightly from $99.9\%$ to $99.7\%$. These experiments demonstrate that MCPS reduces the communication cost with negligible loss of accuracy.

\subsubsection{Handling Flow Dynamics}
The performance of DAPR is presented in Figure~\ref{fig_flowdynamics}. The number of active flows in the $60$s traces varies from $243$ to $1746$. The communication cost of the per-flow polling method is plotted for comparison and the cost is in proportion to the number of active flows. ``Recompute'' method is given as optimal since it is the cost by recomputing the polling scheme every second. Clearly, DAPR does not increase too much communication cost compared with the recompute method. This is because current polling scheme consists many polling all switches, which means most of the new flows have been covered by the current polling scheme. Although there exist plenty of short flows, MCPS can still keep relatively stable. Sometimes, the performance of the heuristic is even better than the recompute method. The reason is that the polling scheme is calculated by an approximation algorithm. Increasing a limited number of single polling has little impact on the total communication cost. Therefore, the scheme obtained by DAPR might be better than the recompute method in a short period of time. Moreover, DAPR with ARI ($\tau=0.3$) is better than DAPR with $10$s recompute interval as it only recomputes $4$ times which is $33.3\%$ less than the fixed interval recomputation. As mentioned, the polling scheme generated by MCPS is stable even the active flows varies a lot. This property motivates us to employ ARI to adaptively increase the recompute interval in practice. In summary, DAPR is able to tackle dramatic flow dynamics in the network. Combined with ARI, it can further reduce the computing overhead of CeMon.
\begin{figure}[!t]
  \centering
  \includegraphics[width=\linewidth]{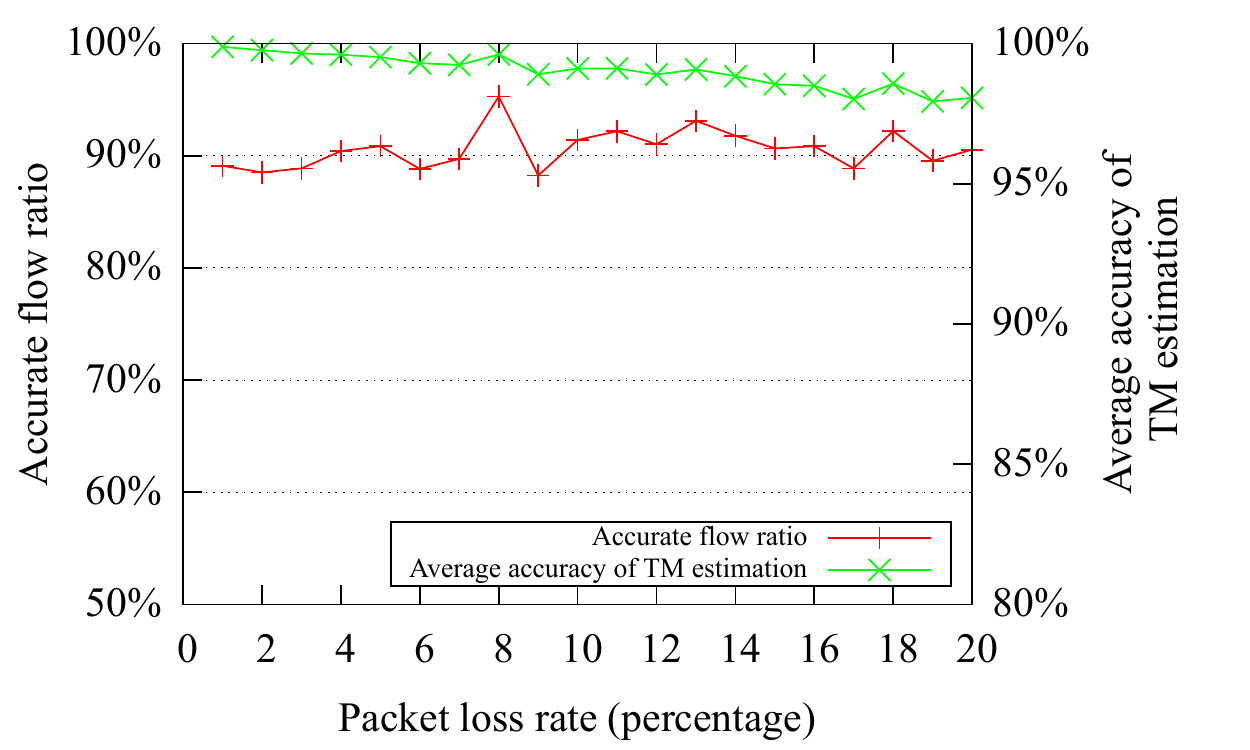}
  \caption{Accurate flow ratio and average accuracy of traffic matrix estimation as the packet loss rate varies from $0$ to $20\%$ (with a loss switch ratio of $10\%$).}
  \label{fig_accuracyvspacketloss}
\end{figure}

\begin{figure}[!t]
  \centering
  \includegraphics[width=\linewidth]{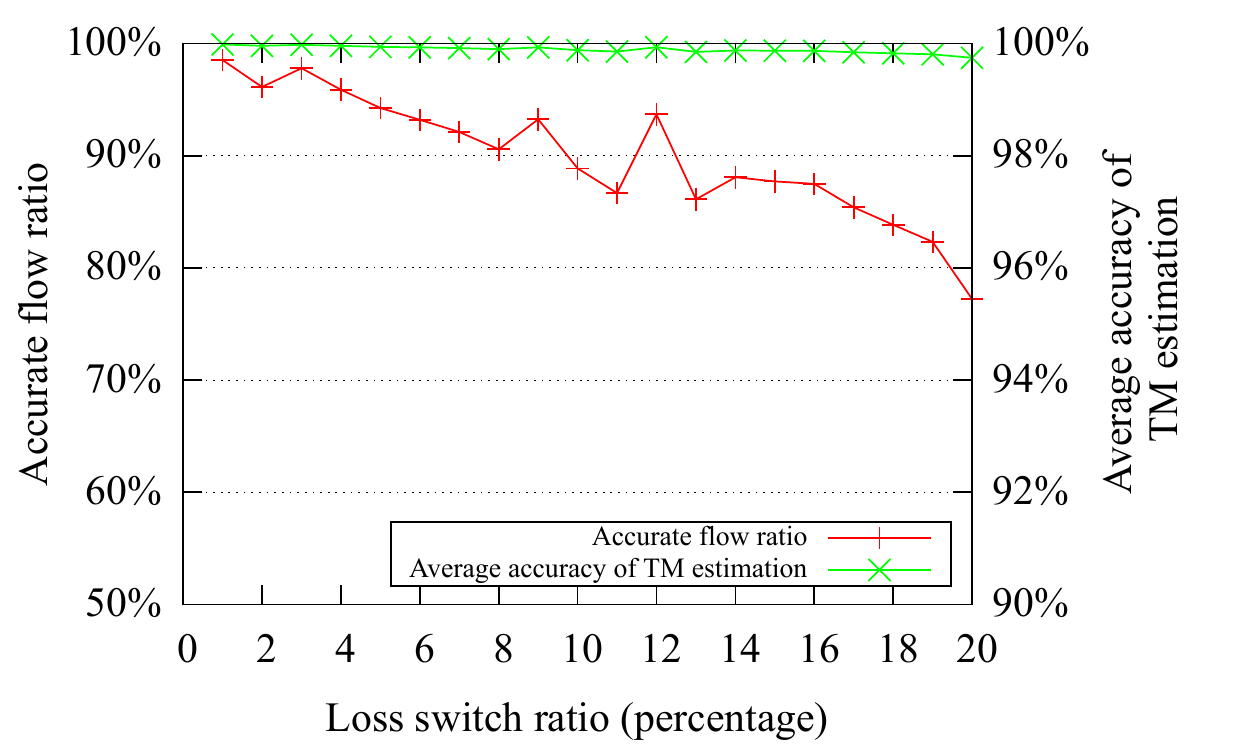}
  \caption{Accurate flow ratio and average accuracy of traffic matrix estimation as the loss switch ratio varies from $0$ to $20\%$ (with a packet loss rate of $1\%$).}
  \label{fig_accuracyvslossswitch}
\end{figure}

\subsection{AFPS Results}
To evaluate the polling overhead and effectiveness of AFPS, we analyze the communication cost and the accuracy by implementing a link utilization task according to Eq.~\ref{eq_lu} on top of AFPS. We measure the link utilization for a class C subnet on a link using real packet traces~\cite{trafficchar}. We utilize both TCP and UDP traffic which have different traffic characteristics to verify the performance of AFPS.
\begin{figure}
  \centering
  \includegraphics[width=\linewidth]{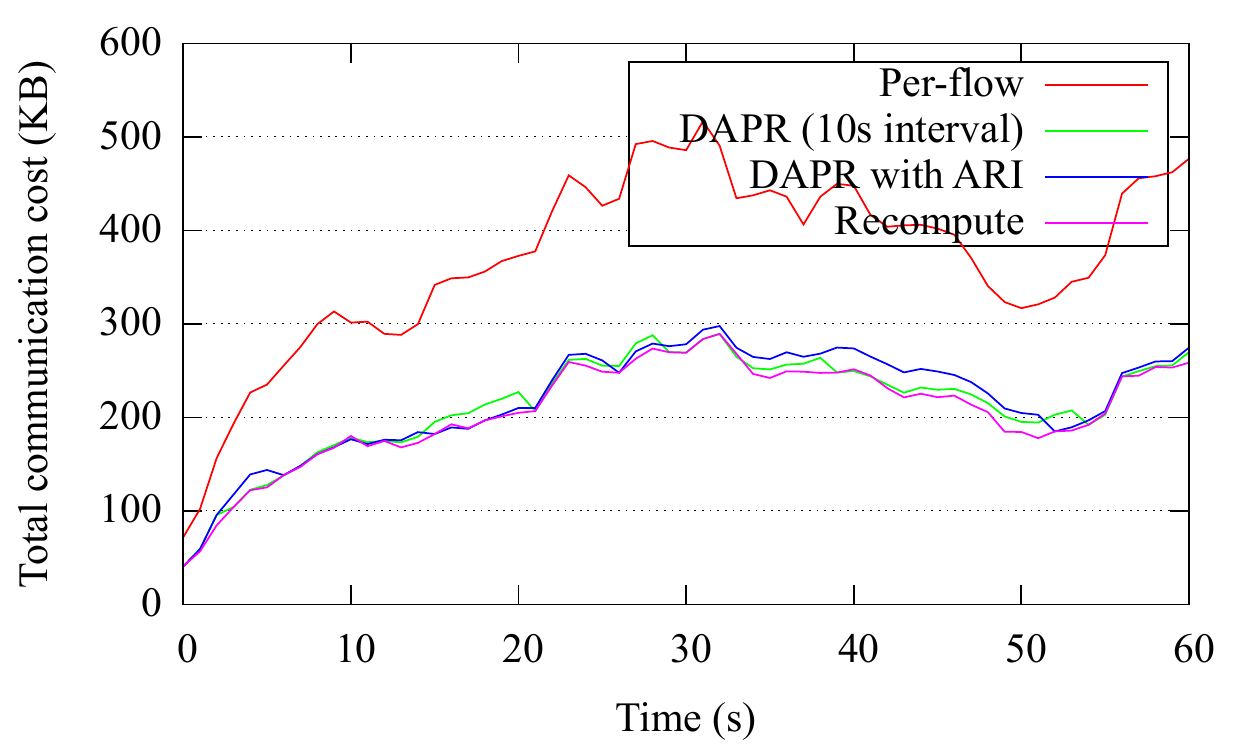}
  \caption{The performance of DAPR.}
  \label{fig_flowdynamics}
\end{figure}

\subsubsection{Accuracy and Communication Cost}
The initial sampling interval for all algorithms are set to $1$s. The link utilization measurement interval is set to $1$s as well. The soft timeout for each flow is set to $10$s which is common used in practice. The minimum and maximum polling intervals are set to $0.5$s and $5$s respectively which is the same as in PayLess~\cite{payless}.

Figure~\ref{fig_afpsutilization} shows the link utilization of different tuning algorithms for TCP traffic. The number of active flows in this $60$s trace during this period is $2668$. The actual link utilization and the periodic polling are plotted for comparison. The actual utilization is highlighted. The link utilization obtained by our tuning algorithms follows the actual utilization closely. However, adaptive sampling methods may miss some traffic spikes. For instance, from $10$s to $15$s, the small traffic peak is not detected by all the sampling methods. Figure~\ref{fig_afpscost} is the corresponding communication cost for different sampling methods. Clearly, PT, EWMAT and SWT generate less sampling messages than the periodic polling. As it is not apparent to tell which method is more accurate, quantitative analysis and comparison for each sampling method will be given in the next section.

Compared with TCP, UDP has no flow control mechanism and usually lasts longer. The number of active flows in this trace is $111$, which is much less than the TCP traffic. Figure~\ref{fig_afpsutilizationudp} and Figure~\ref{fig_afpscostudp} depict the measured link utilization and the corresponding message number for UDP traffic respectively. Even the UDP traffic fluctuates sharply, PT, EWMAT and SWT follow the traffic pattern well. Since UDP has many long flows, AFPS is able to reduce the polling messages significantly comparing with the periodic polling. This is because these tuning algorithms can better follow the flow pattern as they use historical data to predict future traffics. The number of polling requests for both TCP and UDP traffic fluctuates over time since the AFPS is adaptive to the drastic flow dynamics. Comparatively, the volatility of CeMon's tuning algorithms is slightly smaller than periodic polling and PayLess. The algorithms converge and produce better results when these flows last long.
\begin{figure}[!t]
  \centering
  \includegraphics[width=\linewidth]{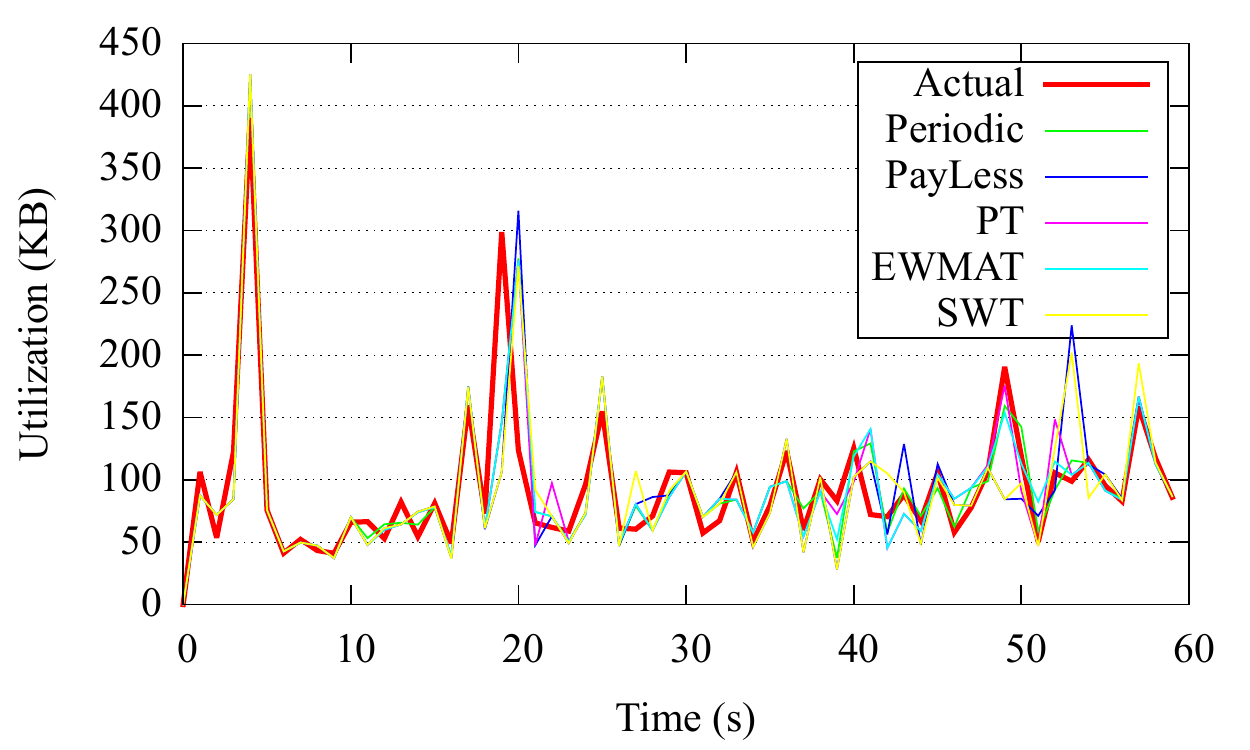}
  \caption{The measured link utilization by AFPS for TCP traffic.}
  \label{fig_afpsutilization}
\end{figure}

\begin{figure}[!t]
  \centering
  \includegraphics[width=\linewidth]{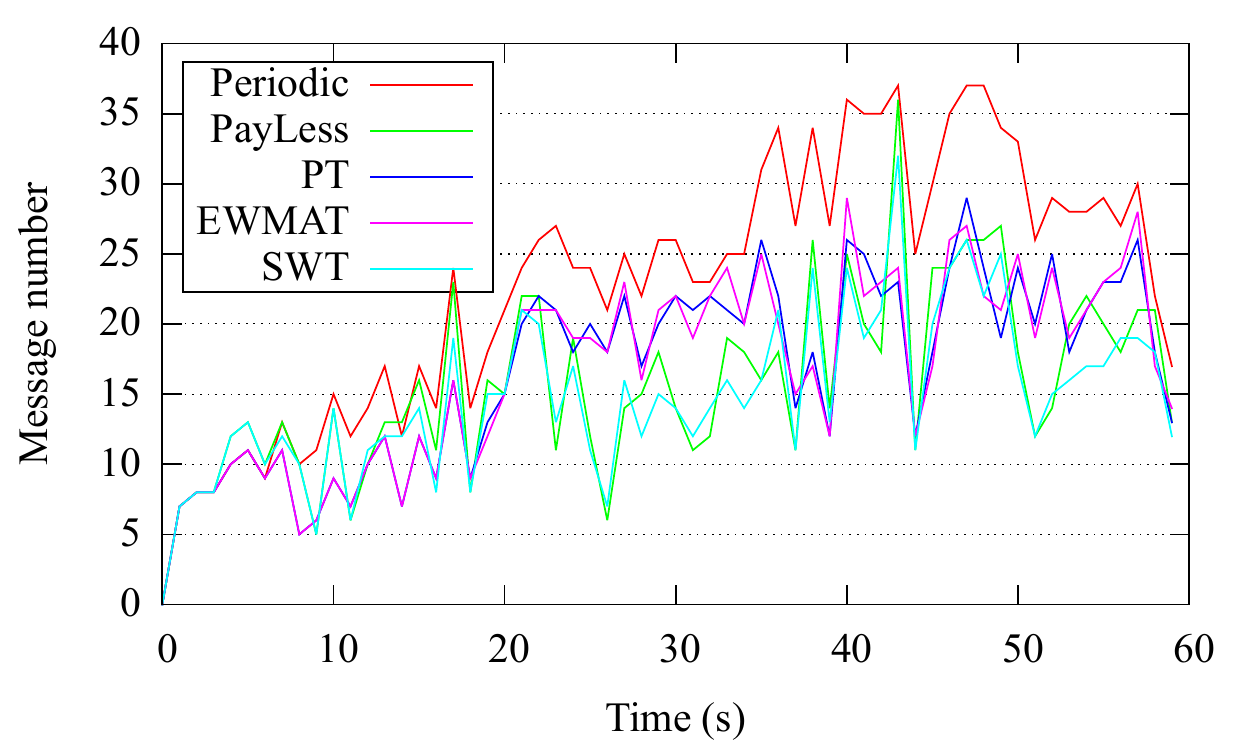}
  \caption{The number of sampling messages for TCP traffic.}
  \label{fig_afpscost}
\end{figure}

\subsubsection{Tuning Algorithms Comparison}
We compare the proposed tuning algorithms with the periodic polling and a prior work PayLess~\cite{payless}. The measurement error is given by Eq.~\ref{eq_error}. Figure~\ref{fig_afpscomp} presents the measurement error and the cost for the tuning algorithms. Obviously, the proposed tuning algorithms for AFPS significantly outperform the periodic polling in terms of the polling cost for both TCP and UDP traffic. In particular, for UDP traffic, PT and EWMAT save up to $67\%$ of the communication cost, which further reduce the cost by $13\%$ compared with PayLess. For TCP traffic which contains plenty of short flows, PT, EWMAT and SWT save $20\%$, $20\%$ and $26\%$ of the polling cost respectively compared with the periodic polling. For the measurement error, both PayLess and AFPS have a slightly larger error than the periodic polling. This is because the periodic polling requests the statistics more aggressively and incurs more polling overhead. However, it is worth noting that SWT and PayLess have nearly the same measurement error while SWT generates less polling messages. Specifically, PT reduces the cost by $20\%$ with only $1.5\%$ accuracy loss. Besides, consider SWT is almost parameter-free and easy to configure, it is superior to other sampling methods in the light of the cost and the accuracy. These results demonstrate that AFPS strikes a good trade-off between the measurement accuracy and the cost. By trading a little accuracy, AFPS notably reduces the polling overhead, which is crucial for fine-grained measurements.
\begin{figure}[!t]
  \centering
  \includegraphics[width=\linewidth]{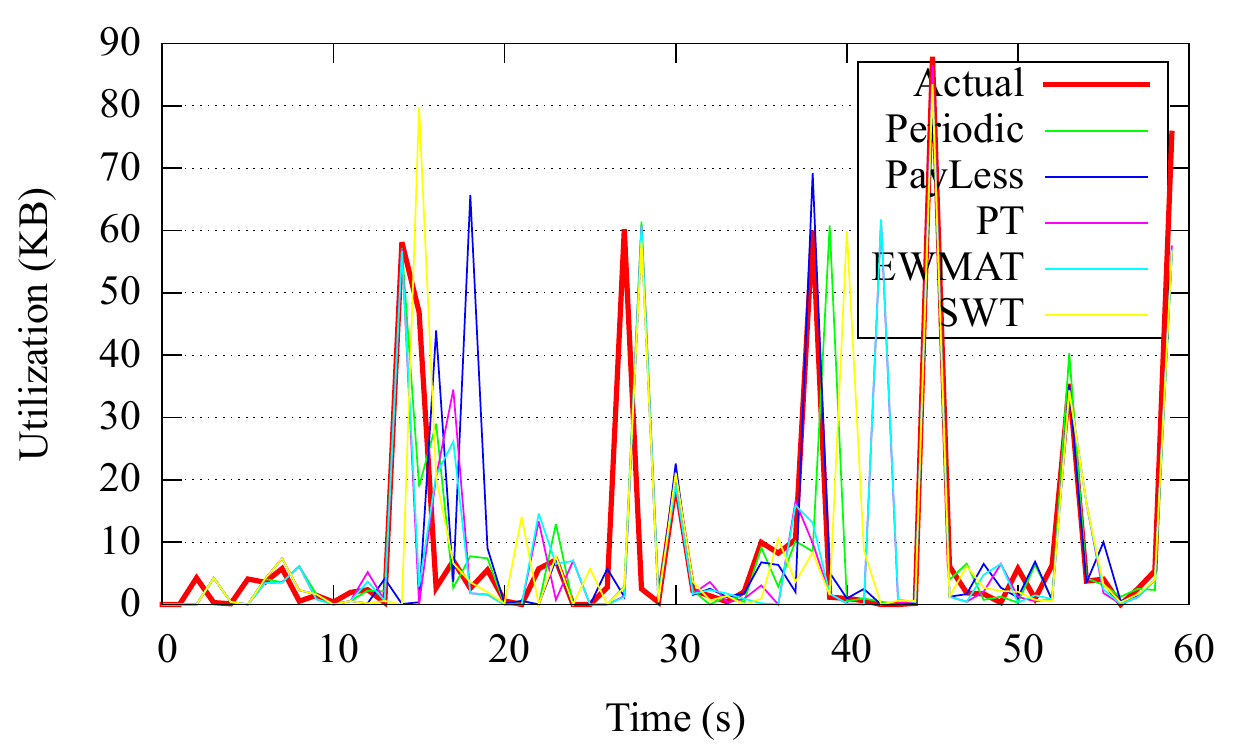}
  \caption{The measured link utilization by AFPS for UDP traffic.}
  \label{fig_afpsutilizationudp}
\end{figure}

\begin{figure}[!t]
  \centering
  \includegraphics[width=\linewidth]{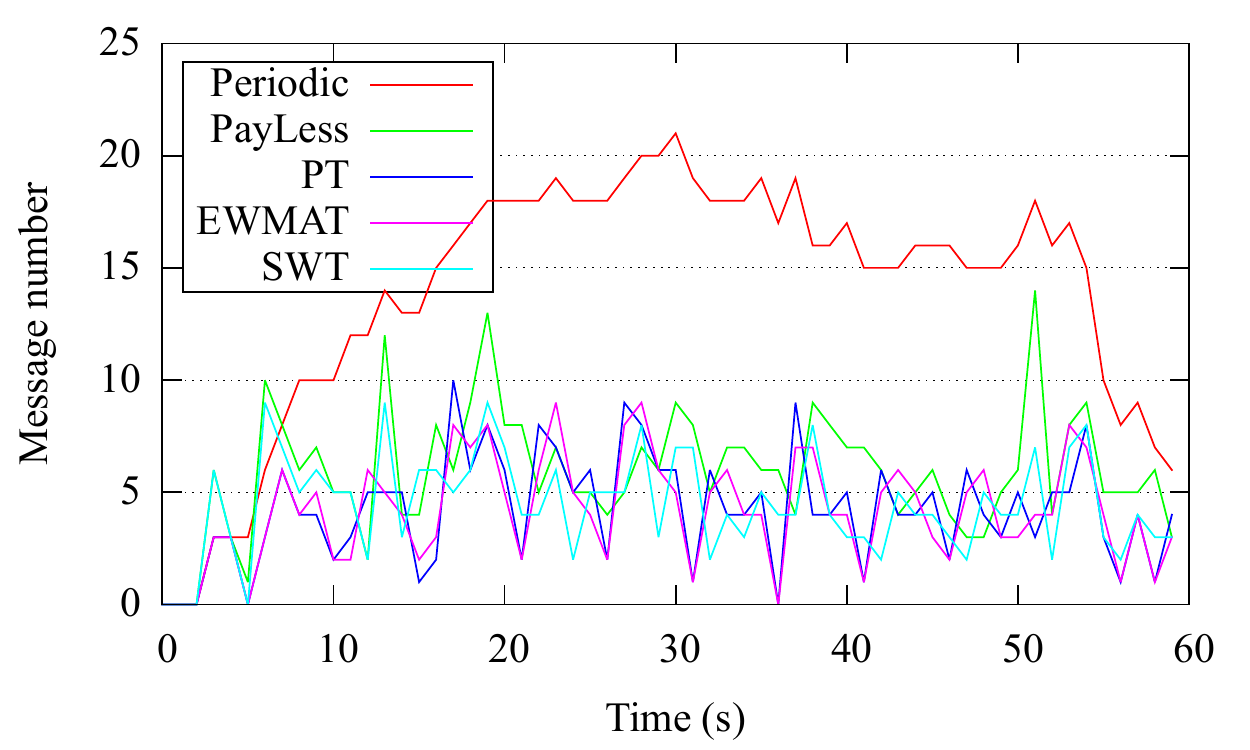}
  \caption{The number of sampling messages for UDP traffic.}
  \label{fig_afpscostudp}
\end{figure}

\section{Related Work} \label{sec_relatedwork}
Prior work explored different approaches to design a low-cost high-accuracy measurement system for SDN-based networks. Our earlier work FlowCover~\cite{flowcover} presented preliminary results of reducing monitoring overhead in out-of-band deployment of SDN. Dynamically changing the aggregation granularity is a common approach to reduce the measurement cost in SDN. L. Jose et al. detected hierarchical heavy hitters by changing the measurement rules in the switches~\cite{onlineaggregate}. OpenWatch~\cite{adaptiveflowcounting} adjusted the measurement granularity in both the spatial and the temporal dimensions. To further reduce the monitoring overhead, FlowSense~\cite{flowsense} presented a push-based method to measure the network link utilization with zero overhead. However, FlowSense can only obtain the link utilization at discrete points in time and cannot meet the real-time monitoring requirement. Besides, Amazon CloudWatch~\cite{cloudwatch} provided APIs to monitor online service status. Planck~\cite{planck} employed oversubscribed port mirroring to gather network states with milliseconds-scale. These work attempted to trade off the accuracy and overhead in the measurement applications by aggregation or sampling. CeMon is orthogonal to these work since it works as a bottom layer to reduce the fetching counter bandwidth consumption.
\begin{figure}[!t]
  \centering
  \includegraphics[width=\linewidth]{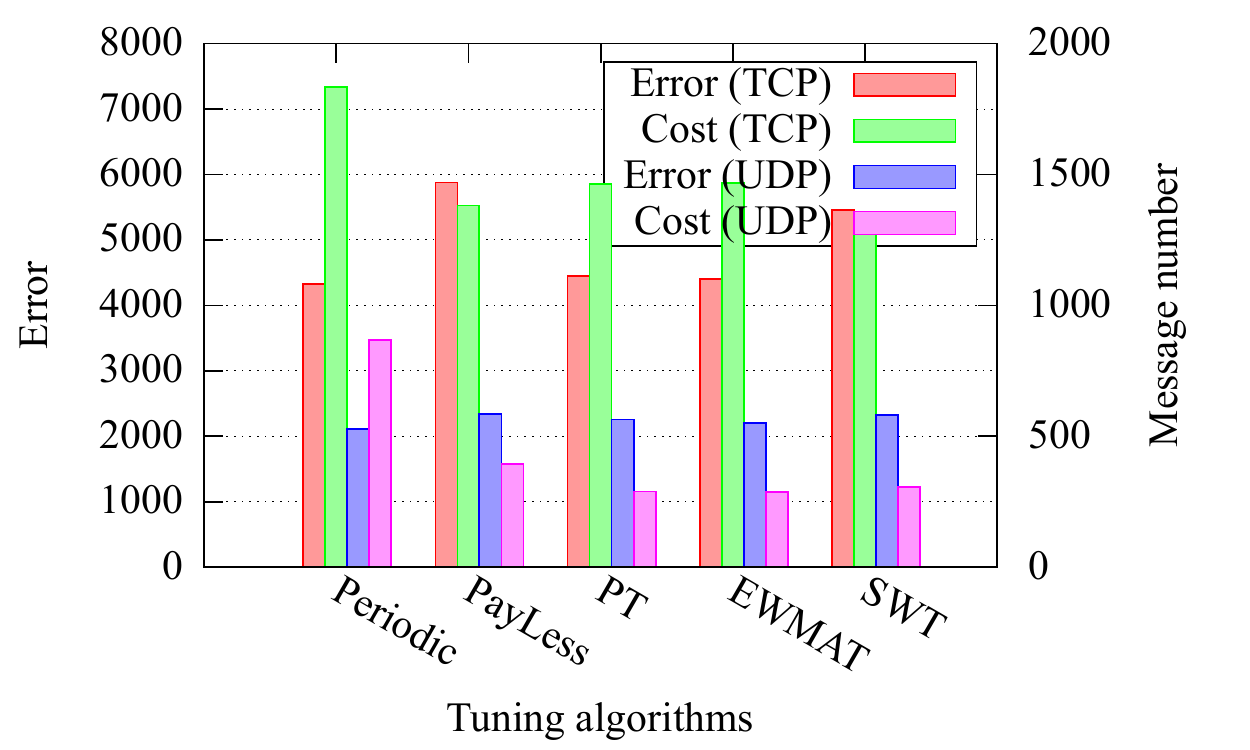}
  \caption{Comparison of different tuning algorithms.}
  \label{fig_afpscomp}
\end{figure}

Sampling is another alternative to alleviate the monitoring overhead. DevoFlow~\cite{devoflow} proposed a sampling-based method to improve the performance of statistics collection. CSAMP~\cite{csamp} maximized the monitoring flow coverage by consistent sampling. Moreover, a sampling extension for monitoring applications is presented in~\cite{flexam}. The most related work to our proposal is OpenTM~\cite{opentm}, which is a traffic matrix estimation system that gathers flow statistics by different querying strategies. However, it collected active flow statistics on a per-flow basis without considering the bandwidth consumption. In contrast, our approach globally optimizes the polling strategy to fetch counters from switches.

Previous work on programmable measurement frameworks has demonstrated the benefits of customized measurement applications. ProgME~\cite{progme} enabled flexible flow counting by defining the concept of flowset that is an arbitrary set of flows for different applications. Another measurement primitive OpenSketch~\cite{opensketch} allowed customized TCAM-based measurement in SDN. DCM~\cite{dcm} provided a two-stage bloom filter switch architecture to facilitate the SDN monitoring. These existing proposals mainly focused on the application layer measurement, whereas our system is a shim layer between the controller and the physical switches. Our statistic collection schemes can be applied to all these approaches to reduce the measurement overhead.

Benefiting from the global visibility of SDN and low memory usage of streaming algorithms, more flexible SDN measurement techniques~\cite{opensketch,onlineaggregate, resourceaccuracy} are proposed to accommodate different sketch-based measurement algorithms. Specifically, OpenSketch presented a SDN switch architecture with a three-stage pipeline, which supported various TCAM-based measurement tasks with low memory usage and high flexibility. These approaches required extra wildcard rules which are stored in TCAMs to implement measurement applications. However, TCAMs are precious resources and only have thousands of entries in today's commodity switches~\cite{devoflow}. The trade-off between the resource consumption and measurement accuracy had first been studied in~\cite{resourceaccuracy}. It also introduced the resource allocation problem in software defined measurement. A follow up work DREAM~\cite{dream} extended the measurement scenario from a single switch to multiple switches. DREAM combined the local and global accuracy estimation to achieve a desired level of accuracy using a practical allocation algorithm.

Reducing the communication cost of distributed systems has attracted the attention of many researchers. Li et al.~\cite{distributedmonitoring} adopted integer linear programming to optimize the distributed monitoring infrastructure in traditional networks. Three heuristics were proposed to reduce the deployment cost of polling nodes. Additionally, many theoretical studies~\cite{communicationcost,datasummary, tongcd12, freqestimation} investigated minimizing the communication cost of ``thresholded counts'' in distributed monitoring systems by setting local thresholds.

\section{Discussion} \label{sec_discussion}
\textbf{Selection of polling schemes.} The MCPS and the AFPS target at different monitoring scenarios. If network operators setup network-wide monitoring application, the MCPS is suitable for this scenario since it collects all the flow statistics in a cost-effective manner. The aggregated polling requests and replies greatly reduce the monitoring overhead. However, if network operators only want to know the flow statistics for a host, the AFPS should be used because the sampling technique provides a light-weight, fine-grained and high-accuracy monitoring result.

\textbf{Proactive rules.} If forwarding rules are installed proactively, CeMon may not track the matching flow status because flow arrive messages will not be sent to the controller. However, if these rules are in the selected polling all switches, CeMon is still able to obtain the flow statistics from them. One possible solution is that network operators explicitly specify these proactive rules in the controller. Another alternative is that CeMon periodically polls all flows in every switch to track all the active flows in the network.

\textbf{Wildcard rules.} Fine-grained analysis on each wildcard rule is able to further reduce the polling overhead. For instance, if an application monitors the link utilization of a subnet and a corresponding forwarding wildcard rule exists, then a single polling of this rule is the optimal solution. We leave such extension to our future work.

\section{Conclusion} \label{sec_conclusion}
In this paper, we propose CeMon, a low-cost high-accuracy SDN monitoring system. We analyze the communication overhead of SDN monitoring and propose two novel generic polling schemes to accommodate various monitoring applications. Specifically, MCPS globally optimizes the polling overhead to gather all flow statistics. Heuristics are presented to generate the polling scheme efficiently and handle flow dynamics. AFPS are proposed as a complementary method to collect statistics from a subset of active flows. Despite the uniform flow level measurements formulation, three adaptive algorithms are presented to dynamically adjust polling intervals to further reduce the monitoring overhead. Both emulation and simulation results show that MCPS reduces more than $50\%$ of the communication cost. In addition, we use real packet traces to demonstrate that AFPS significantly reduces the monitoring overhead with negligible loss in accuracy.

\section*{Acknowledgments}
This research is supported by HKUST Research Grants Council (RGC) 613113.

\section*{References}
\bibliographystyle{elsarticle-num}
\bibliography{references}

\begin{thebibliography}{10}
\expandafter\ifx\csname url\endcsname\relax
  \def\url#1{\texttt{#1}}\fi
\expandafter\ifx\csname urlprefix\endcsname\relax\def\urlprefix{URL }\fi
\expandafter\ifx\csname href\endcsname\relax
  \def\href#1#2{#2} \def\path#1{#1}\fi

\bibitem{hedera}
M.~Al-Fares, S.~Radhakrishnan, B.~Raghavan, N.~Huang, A.~Vahdat, Hedera:
  dynamic flow scheduling for data center networks, in: NSDI, 2010.

\bibitem{helios}
N.~Farrington, G.~Porter, S.~Radhakrishnan, H.~H. Bazzaz, V.~Subramanya,
  Y.~Fainman, G.~Papen, A.~Vahdat, Helios: a hybrid electrical/optical switch
  architecture for modular data centers, in: SIGCOMM, 2010.

\bibitem{netflow}
{NetFlow},
  \url{http://www.cisco.com/c/en/us/products/ios-nx-os-software/ios-netflow/index.html}.

\bibitem{sflow}
M.~Wang, B.~Li, Z.~Li, {sFlow: Towards resource-efficient and agile service
  federation in service overlay networks}, in: ICDCS, 2004.

\bibitem{opensketch}
M.~Yu, L.~Jose, R.~Miao, {Software defined traffic measurement with
  OpenSketch}, in: NSDI, 2013.

\bibitem{reformulateplacement}
G.~R. Cantieni, G.~Iannaccone, C.~Barakat, C.~Diot, P.~Thiran, Reformulating
  the monitor placement problem: Optimal network-wide sampling, in: CISS, 2006.

\bibitem{progme}
L.~Yuan, C.-N. Chuah, P.~Mohapatra, {ProgME}: towards programmable network
  measurement, ToN.

\bibitem{dcm}
Y.~Yu, Q.~Chen, X.~Li, Distributed collaborative monitoring in software defined
  networks, in: HotSDN, 2014.

\bibitem{onlineaggregate}
L.~Jose, M.~Yu, J.~Rexford, Online measurement of large traffic aggregates on
  commodity switches, in: HotICE, 2011.

\bibitem{opentm}
A.~Tootoonchian, M.~Ghobadi, Y.~Ganjali, {OpenTM: traffic matrix estimator for
  OpenFlow networks}, in: PAM, 2010.

\bibitem{mahout}
A.~R. Curtis, K.~Wonho, P.~Yalagandula, Mahout: Low-overhead datacenter traffic
  management using end-host-based elephant detection, in: INFOCOM, 2011.

\bibitem{flowsense}
C.~Yu, C.~Lumezanu, Y.~Zhang, V.~Singh, G.~Jiang, H.~V. Madhyastha, {FlowSense:
  monitoring network utilization with zero measurement cost}, in: PAM, 2013.

\bibitem{openflow}
N.~McKeown, T.~Anderson, H.~Balakrishnan, G.~Parulkar, L.~Peterson, J.~Rexford,
  S.~Shenker, J.~Turner, {OpenFlow}: enabling innovation in campus networks,
  SIGCOMM CCR.

\bibitem{openflowspec10}
Openflow switch specification 1.0.0,
  \url{https://www.opennetworking.org/images/stories/downloads/sdn-resources/onf-specifications/openflow/openflow-spec-v1.0.0.pdf}.

\bibitem{inbandfastrecovery}
S.~Sharma, D.~Staessens, D.~Colle, M.~Pickavet, P.~Demeester, Fast failure
  recovery for in-band openflow networks, in: Design of Reliable Communication
  Networks, 2013.

\bibitem{dream}
M.~Moshref, M.~Yu, R.~Govindan, A.~Vahdat, {DREAM}: Dynamic resource allocation
  for software-defined measurement, in: SIGCOMM, 2014.

\bibitem{payless}
S.~R. Chowdhury, M.~F. Bari, R.~Ahmed, R.~Boutaba, {PayLess: A Low Cost Network
  Monitoring Framework for Software Defined Networks}, in: NOMS, 2014.

\bibitem{flowsizedistribution}
A.~Kumar, M.~Sung, J.~J. Xu, J.~Wang, Data streaming algorithms for efficient
  and accurate estimation of flow size distribution, in: SIGMETRICS, 2004.

\bibitem{adaptiveflowcounting}
Y.~Zhang, An adaptive flow counting method for anomaly detection in sdn, in:
  CoNEXT, 2013.

\bibitem{approximationbook}
V.~V. Vazirani, Approximation algorithms, Springer-Verlag New York, Inc., 2001.

\bibitem{datasummary}
G.~Cormode, S.~Muthukrishnan, An improved data stream summary: the count-min
  sketch and its applications, Journal of Algorithms.

\bibitem{ewma}
C.~A. Lowry, W.~H. Woodall, C.~W. Champ, S.~E. Rigdon, A multivariate
  exponentially weighted moving average control chart, Technometrics.

\bibitem{aimd}
D.-M. Chiu, R.~Jain, Analysis of the increase and decrease algorithms for
  congestion avoidance in computer networks, Computer Networks and ISDN
  Systems.

\bibitem{topologyzoo}
S.~Knight, H.~Nguyen, N.~Falkner, R.~Bowden, M.~Roughan, The internet topology
  zoo, JSAC.

\bibitem{erdos}
B.~Bollob\'{a}s, Random graphs, Academic Press, 1985.

\bibitem{waxman}
B.~M. Waxman, Routing of multipoint connections, JSAC.

\bibitem{pox}
{POX} controller, \url{http://www.noxrepo.org/pox/about-pox/}.

\bibitem{mininet}
B.~Lantz, B.~Heller, N.~McKeown, A network in a laptop: Rapid prototyping for
  software-defined networks, in: HotNets, 2010.

\bibitem{openvswitch}
{Open vSwitch}, \url{http://openvswitch.org/}.

\bibitem{trafficchar}
T.~Benson, A.~Akella, D.~A. Maltz, Network traffic characteristics of data
  centers in the wild, in: IMC, 2010.

\bibitem{flowcover}
Z.~Su, T.~Wang, Y.~Xia, M.~Hamdi, Flowcover: Low-cost flow monitoring scheme in
  software defined networks, in: GLOBECOM, 2014.

\bibitem{cloudwatch}
Amazon cloudwatch, \url{http://aws.amazon.com/cloudwatch/}.

\bibitem{planck}
J.~Rasley, B.~Stephens, C.~Dixon, E.~Rozner, W.~Felter, K.~Agarwal, J.~Carter,
  R.~Fonseca, Planck: millisecond-scale monitoring and control for commodity
  networks, in: SIGCOMM, 2014.

\bibitem{devoflow}
A.~R. Curtis, J.~C. Mogul, J.~Tourrilhes, P.~Yalagandula, P.~Sharma,
  S.~Banerjee, {DevoFlow: scaling flow management for high-performance
  networks}, in: SIGCOMM, 2011.

\bibitem{csamp}
V.~Sekar, M.~K. Reiter, W.~Willinger, H.~Zhang, R.~R. Kompella, D.~G. Andersen,
  {CSAMP}: a system for network-wide flow monitoring, in: NSDI, 2008.

\bibitem{flexam}
S.~Shirali-Shahreza, Y.~Ganjali, {FleXam}: flexible sampling extension for
  monitoring and security applications in openflow, in: HotSDN, 2013.

\bibitem{resourceaccuracy}
M.~Moshref, M.~Yu, R.~Govindan, {Resource/accuracy tradeoffs in
  software-defined measurement}, in: HotSDN, 2013.

\bibitem{distributedmonitoring}
L.~Li, M.~Thottan, B.~Yao, S.~Paul, Distributed network monitoring with bounded
  link utilization in ip networks, in: INFOCOM, 2003.

\bibitem{communicationcost}
R.~Keralapura, G.~Cormode, J.~Ramamirtham, Communication-efficient distributed
  monitoring of thresholded counts, in: SIGMOD, 2006.

\bibitem{tongcd12}
T.~Yongxin, C.~Lei, D.~Bolin, {Discovering Threshold-based Frequent Closed
  Itemsets over Probabilistic Data}, in: ICDE, 2012.

\bibitem{freqestimation}
Z.~Huang, K.~Yi, Y.~Liu, G.~Chen, Optimal sampling algorithms for frequency
  estimation in distributed data, in: INFOCOM, 2011.

\end{thebibliography}

\end{document}